\documentclass[12pt]{article}

\usepackage[round]{natbib}
\usepackage{amsmath,amssymb,amsthm,mathtools,bm}
\usepackage{bbm}         
\usepackage{mathrsfs}    
\usepackage{enumitem}    

\usepackage{graphicx}
\usepackage{booktabs}
\usepackage{multirow}
\usepackage{rotating}

\usepackage{microtype}
\allowdisplaybreaks

\usepackage{xcolor}
\usepackage{hyperref}

\usepackage{algorithm}
\usepackage{algpseudocode}

\usepackage{authblk} 

\newtheorem{lemma}{Lemma}
\newtheorem{proposition}{Proposition}
\newtheorem{assumption}{Assumption}

\title{Smart Contract Adoption in Derivative Markets under Bounded Risk:\\
An Optimization Approach}

\author[1]{Jinho Cha}
\author[2]{Long Pham}
\author[3]{Thi Le Hoa Vo\thanks{Corresponding author: thi-le-hoa.vo@univ-rennes.fr}}
\author[4]{Jaeyoung Cho}
\author[5]{Jaejin Lee}

\affil[1]{Department of Computer Science, Gwinnett Technical College, Lawrenceville, GA, USA}
\affil[2]{Department of Decision Sciences and Economics, Texas A\&M University--Corpus Christi, TX, USA}
\affil[3]{IGR-IAE Rennes, Universit\'e de Rennes, CREM UMR CNRS 6211, Rennes, France}
\affil[4]{Department of Computer Science, Prairie View A\&M University, Prairie View, TX, USA}
\affil[5]{Intel Corporation, Chandler, AZ, USA}

\date{} 

\begin{document}

\maketitle

\begin{abstract}
This study develops and analyzes an optimization model of smart contract adoption under bounded risk, 
linking structural theory with simulation and real-world validation. 
We examine how adoption intensity and profitability respond to threshold costs, variance shocks, and 
distributional changes. Comparative statics and Monte Carlo experiments show that adoption intensity 
$\alpha^\star$ is structurally pinned at a boundary solution, invariant to variance and heterogeneity, 
while profitability and service outcomes are variance-fragile, eroding under volatility and heavy-tailed demand. 
A sharp threshold in the fixed-cost parameter $A_3$ triggers discontinuous adoption collapse (H1), 
variance shocks reduce profits monotonically but not adoption (H2), and additional results on readiness 
heterogeneity (H3), profit--service co-benefits (H4), and distributional robustness (H5) confirm the duality 
between stable adoption and fragile payoffs. External validity checks further establish convergence of 
sample-average approximation at the canonical $\mathcal{O}(1/\sqrt{N})$ rate (H6). Empirical validation using 
S\&P500 returns and the MovieLens-100K dataset corroborates the theoretical structure: bounded and heavy-tailed 
distributions fit better than Gaussian models, and profits diverge across volatility regimes even as adoption 
remains stable. Taken together, the results demonstrate that adoption choices are robust to uncertainty, 
but their financial consequences are highly fragile. For operations and finance, this duality underscores the 
need for risk-adjusted performance evaluation, option-theoretic modeling, and distributional stress testing in 
strategic investment and supply chain design.
\end{abstract}

\textbf{Keywords:} Smart Contracts; Bounded Risk; Derivative Markets; Optimization; Threshold Effects; Variance Fragility; Robustness


\section{Introduction}

Derivative markets have expanded substantially in recent years, with structured products and risk-transfer instruments becoming increasingly common. These instruments allow market participants to hedge exposures and share risks more effectively, yet they also raise persistent concerns about transparency and counterparty exposure. Several financial episodes have illustrated how limited contract enforcement or monitoring can amplify systemic stress, underscoring the need for more reliable settlement mechanisms \citep{bolton2005contract,cachon2003supply}. In this context, digital technologies such as blockchain and smart contracts have been proposed as potential tools to enhance contract reliability and reduce frictions in clearing and settlement \citep{catalini2016economics,cong2019blockchain,schar2021defi}.

The idea of embedding smart contracts into derivative transactions is appealing, but adoption remains relatively limited. Market participants face a fundamental trade-off: efficiency and transparency gains must be weighed against the costs of technological implementation, integration with legacy systems, and coordination across counterparties \citep{tapscott2017blockchain}. Moreover, most existing modeling approaches assume heavy-tailed or unbounded risk distributions, which emphasize extreme events but may overstate volatility in settings where exposures are naturally capped. In practice, mechanisms such as margin rules, collateral requirements, and regulatory limits frequently impose bounds on both losses and gains \citep{embrechts1997extremal,johnson1995continuous}. 

To capture such bounded exposures in a more realistic manner, we adopt the lognormal distribution. Lognormal distributions are strictly non-negative and positively skewed, making them particularly suitable for modeling derivative exposures where downside losses are limited but upside shocks can be extreme. This property provides a natural representation of demand and settlement risk in financial markets, where asymmetry and fat tails are empirically observed \citep{embrechts1997extremal}. Recent studies also emphasize the importance of distributional choice in variance and tail risk modeling across asset classes, reinforcing the need for flexible yet tractable specifications \citep{fathi2025variance}. Alongside robustness checks using beta distributions, this choice allows us to examine adoption incentives under realistic bounded risk environments.

The purpose of this paper is to develop a tractable optimization framework for smart contract adoption under bounded risk. We introduce an adoption intensity decision variable and formulate the associated profit maximization problem under lognormal demand. The model remains convex and analytically interpretable, which allows us to characterize optimality conditions and highlight threshold effects in adoption. In particular, we show how adoption collapses or accelerates depending on the interplay between adoption costs and supplier digital readiness. 

Our analysis proceeds in two steps. First, we establish the baseline convex optimization model and derive conditions for positive adoption intensity. Second, we conduct numerical experiments using synthetic data, including sensitivity analysis with lognormal and beta demand distributions \citep{shapiro2014stochastic,rockafellar2000cvar}. This design allows us to provide comparative statics and counterfactual scenarios, including cost shocks and risk shocks, that speak directly to managerial and policy questions. 

Overall, our contribution is incremental but useful. We extend the contract theory and digital adoption literature by embedding lognormal risk explicitly into the adoption decision, and by showing how threshold effects and supplier heterogeneity shape equilibrium outcomes. While our framework is stylized and based on synthetic data, it offers a transparent and interpretable foundation that can inform both future empirical validation and regulatory discussions of digital adoption in derivative markets. 

Finally, based on this motivation, we formulate a set of testable hypotheses (Section~\ref{sec:hypotheses}) that directly link adoption cost, risk variance, supplier readiness, and distributional assumptions to observable adoption outcomes. These hypotheses guide our experimental design and provide a rigorous structure for interpreting numerical results.

\section{Literature Review}

\subsection{Financial Contract Design}

The study of financial contract design has long emphasized mechanisms to align incentives and distribute risk among contracting parties. Classical models consider arrangements such as buyback contracts, revenue-sharing agreements, and quantity-flexibility provisions \citep{cachon2003supply,lariviere2001selling}. These frameworks typically assume enforcement is guaranteed through legal systems, reputation, or repeated interaction \citep{bolton2005contract,tirole1999incomplete}. More recent work has extended contract theory into supply chains and financial markets, exploring issues of renegotiation, moral hazard, and dynamic incentives \citep{narayanan2004aligning,bernheim1987coalition,maskin2001implementation}. 

A parallel stream of literature has examined incomplete contracts and relational enforcement, emphasizing the role of trust and long-term cooperation \citep{grossman1986costs,hart1990property,baker2002relational}. These insights are particularly relevant for modern markets where counterparties are diverse, information asymmetry is prevalent, and enforcement may be imperfect. 

Building on these foundations, recent studies have started to integrate digital execution and smart contracting into the analysis of contract design. For instance, \citet{atalay2019blockchain} examine how blockchain-enabled settlement alters counterparty incentives, while \citet{harvey2021defi} and \citet{broby2022financial} highlight how decentralized finance (DeFi) changes the enforcement environment by embedding rules directly into code. Very recent work such as \citet{fanti2022blockchain}, \citet{chen2023smart}, and \citet{wang2025contract} stresses the importance of bounded enforcement costs and digital readiness when contracts are executed on distributed ledgers. These contributions suggest that the adoption of smart contracts requires rethinking classical enforcement assumptions, especially in settings where risk exposures are capped by regulation or collateral. 

Our work connects to this evolving stream by extending the classical contract design perspective into a digital execution setting with explicit bounded risk. In particular, we study how the introduction of smart contracts reshapes enforcement and monitoring, while also interacting with supplier heterogeneity and adoption thresholds.

\subsection{Blockchain and Smart Contracts}

The emergence of blockchain technologies introduces the possibility of programmable enforcement. Smart contracts can automate settlement, reduce transaction costs, and improve transparency \citep{catalini2016economics,cong2019blockchain}. Early theoretical work emphasizes their ability to mitigate information asymmetry and reduce reliance on centralized intermediaries \citep{tapscott2017blockchain}. 

Yet, adoption outcomes remain mixed. Case studies and sectoral analyses indicate substantial heterogeneity: while some industries have realized efficiency gains, others report limited improvements once integration and compliance costs are considered \citep{schar2021defi,arora2021blockchain}. Empirical work in financial markets highlights that integrating blockchain into trading and derivatives settlement is non-trivial, as high implementation costs and strategic concerns frequently limit adoption \citep{arora2021blockchain}. 

Recent surveys and comparative studies reinforce this point, noting that coordination costs, digital readiness, and regulatory uncertainty are persistent barriers to scaling smart contract use in finance \citep{shaiku2025blockchain,fanti2022blockchain}. This evidence suggests that adoption decisions are highly context-dependent and cannot be explained by efficiency gains alone. Instead, digital adoption must be analyzed jointly with risk exposure, supplier heterogeneity, and bounded enforcement environments—precisely the dimensions addressed in our optimization framework.

\subsection{Risk Modeling in Financial Markets}

A second relevant stream concerns risk modeling. Classical portfolio and contract studies often rely on normal or heavy-tailed distributions such as Pareto or lognormal to capture extreme outcomes \citep{embrechts1997extremal,mandelbrot1963variation}. While useful for systemic risk analyses, these unbounded models may overstate volatility where contractual or regulatory features cap exposures \citep{johnson1995continuous,glasserman2005importance}. In practice, mechanisms such as margin rules, collateral requirements, and regulatory limits impose upper and lower bounds on realized payoffs \citep{duffie2010dark,brunnermeier2014risk}. 

More recent work examines variance risk and tail properties across multiple asset classes, emphasizing how distributional choice affects inference about shocks and systemic fragility \citep{andersen2015risk,barndorff2008multivariate,fathi2025variance}. This literature highlights that modeling assumptions about boundedness and skewness can meaningfully change conclusions about optimal hedging and adoption incentives. 

Among alternative specifications, lognormal distributions are particularly relevant because they are strictly non-negative and positively skewed, matching empirical features of derivative exposures where downside is limited but upside shocks can be large \citep{embrechts1997extremal,cont2001empirical}. Beta distributions also offer flexibility in representing bounded outcomes and have been applied in option pricing and risk-limiting contracts \citep{carmona2014financial,stoyanov2011fat}. 

Yet these bounded models have been rarely integrated into the study of digital contract adoption, leaving a gap that our work addresses by embedding lognormal and beta risk explicitly into the optimization of smart contract adoption.

\subsection{Positioning of This Study}

Our work builds on these literatures in three complementary ways. 

First, relative to the financial contract design literature, we extend the analysis of adoption incentives into a digital execution environment, where enforcement is programmed rather than legally adjudicated. This shift allows us to study how classical incentive alignment results change when contracts are executed by smart contracts under bounded risk. 

Second, relative to blockchain adoption studies, we move beyond descriptive or case-based accounts by providing an analytical optimization framework that explicitly incorporates adoption costs, supplier readiness, and threshold effects. This enables a sharper characterization of when adoption collapses or accelerates. 

Third, relative to risk modeling studies, we embed bounded lognormal risk directly into the adoption decision. Unlike prior work that focused either on incentive alignment or on digital trust, our framework highlights how convex optimization under bounded risk produces threshold adoption behavior, with clear implications for managerial decision-making and regulatory design.

\section{Research Hypotheses}
\label{sec:hypotheses}
Drawing on the theoretical formulation and literature synthesis, we propose six testable hypotheses that guide our numerical experiments and empirical validation:

\begin{description}
\item[H1 (Threshold).] As adoption cost decreases, adoption intensity $\alpha$ exhibits a sharp jump at a critical threshold, rather than a smooth increase. 
\emph{(Tested by varying $A_3$ and identifying discontinuities in $\alpha^\star$.)}

\item[H2 (Risk–Adoption).] As demand variance $\sigma$ increases, the optimal adoption level $\alpha^\star$ increases monotonically. 
\emph{(Tested by simulating lognormal demand with alternative variance parameters.)}

\item[H3 (Readiness).] Higher average or variance in supplier readiness $\beta_i$ significantly increases $\alpha^\star$ and expected profit. 
\emph{(Tested by varying the support of $\beta_i$ across replications.)}

\item[H4 (Service co-benefit).] Moderate adoption levels simultaneously improve profit and service performance (Fill Rate). 
\emph{(Tested by computing profit–Fill Rate trade-off curves.)}

\item[H5 (Distribution).] Even with the same mean and variance, heavy- vs.\ light-tailed distributions significantly shift the threshold location of $\alpha^\star$. 
\emph{(Tested by comparing lognormal and beta demand distributions.)}

\item[H6 (External validity).] SAA/Monte Carlo approximations converge at rate $O(1/\sqrt{N})$ and preserve adoption thresholds as $N$ grows. 
\emph{(Tested by varying sample size $N$ and tracking the stability of $\alpha^\star$ and profit rankings.)}
\end{description}

\begin{sidewaystable}[htbp]
\centering
\footnotesize
\caption{Dependency of assumptions across hypotheses (H1–H6), grouped by section}
\label{tab:assumption-dependency}
\renewcommand{\arraystretch}{1.25}
\setlength{\tabcolsep}{4pt}
\begin{tabular}{lccccccc l}
\toprule
\textbf{Assumption} 
& \rotatebox{70}{H1 Threshold} 
& \rotatebox{70}{H2 Variance} 
& \rotatebox{70}{H3 Readiness} 
& \rotatebox{70}{H4 Service Co-benefit} 
& \rotatebox{70}{H5 Distribution} 
& \rotatebox{70}{H6 External Validity} 
& \textbf{Notes} \\
\midrule
\multicolumn{8}{l}{\textbf{Model Assumptions (Section~\ref{sec:model})}} \\
Assump.~1 (Demand regularity) & $\checkmark$ & $\checkmark$ & $\checkmark$ & $\checkmark$ & $\checkmark$ & $\checkmark$ & Continuous, finite mean demand (lognormal/beta) \\
Assump.~2 (Concavity and boundedness) & $\checkmark$ & $\checkmark$ & $\checkmark$ & $\checkmark$ & $\checkmark$ & $\checkmark$ & Convexity, bounded optimization \\
Assump.~3 (Strict monotonicity) & $\checkmark$ &  &  &  & $\checkmark$ & $\checkmark$ & Ensures unique $Q^\star$, threshold stability \\
\midrule
\multicolumn{8}{l}{\textbf{Synthetic Data and Simulation (Section~\ref{sec:synthetic-data})}} \\
Lognormal demand (baseline) & $\checkmark$ & $\checkmark$ &  & $\checkmark$ &  &  & Captures bounded skewed exposures \\
Beta demand (robustness) &  &  &  &  & $\checkmark$ &  & Distributional robustness check \\
Monte Carlo approximation (SAA) &  &  &  &  &  & $\checkmark$ & Convergence of simulation estimates \\
\bottomrule
\end{tabular}
\end{sidewaystable}

\section{Model Formulation}
\label{sec:model}

\subsection{Decision Variables}

\subsubsection{Adoption Intensity}
Let $\alpha \in [0,1]$ denote the proportion of smart-contract adoption. 
Here, $\alpha=0$ corresponds to no adoption, while $\alpha=1$ indicates full adoption.

\subsubsection{Order Quantity}
Let $q_i \ge 0$ be the order from supplier $i$, and let the aggregate order be 
$Q=\sum_i q_i$. The firm chooses both $(q_i)$ and $\alpha$ simultaneously.

\subsection{Objective Function}

We adopt a newsvendor-style expected profit formulation with salvage value $s\ge 0$ for overage and shortage penalty $r\ge 0$:
\begin{align}
\max_{\alpha\in[0,1],\,q\ge 0}\ \Pi(\alpha,q)
&= \mathbb{E}\!\Big[\,p\,\min(Q,D)\;+\;s\,(Q-D)^+\;-\;r\,(D-Q)^+\,\Big] \nonumber\\
&\quad - \sum_i c(\alpha,\beta_i)\,q_i\;-\;\psi(\alpha).
\label{eq:obj}
\end{align}

\paragraph{Standing assumptions.}
We impose three assumptions to ensure the problem is well-posed.

\begin{assumption}[Demand regularity]
The demand $D$ is continuous with cumulative distribution function $F$ and density $f$, and satisfies $\mathbb{E}[D]<\infty$. 
This condition holds for standard families such as lognormal or truncated lognormal, which we adopt in simulations.
\end{assumption}

\begin{assumption}[Concavity and boundedness]
Assume $s\le p+r$ and
\begin{equation}
s \ \le\ \underline c
\ :=\ \inf_{\alpha\in[0,1]}\min_{i} c(\alpha,\beta_i)
\ =\ \min_{i}\big(c_{0i}-A_1 - A_2\,\beta_i\big).
\end{equation}
The first inequality ensures that revenue is concave in $Q$, while the second prevents ``salvage arbitrage'' 
(i.e., infinite ordering). Together, these guarantee the optimization problem is bounded and well-posed.
\end{assumption}

\begin{assumption}[Strict monotonicity for uniqueness]
$F$ is strictly increasing, which is satisfied for lognormal and other continuous demand distributions. 
If demand is truncated at an upper bound $b$, then subgradient conditions apply at $Q=b$; 
all subsequent results remain valid using the natural subdifferential.
\end{assumption}

\paragraph{Concavity in \texorpdfstring{$(Q,\alpha)$}{(Q,alpha)}.}
Define the marginal expected revenue
\[
M(Q) \;=\; \frac{\partial}{\partial Q}\,
\mathbb{E}\!\Big[\,p\min(Q,D)+s(Q-D)^+-r(D-Q)^+\,\Big].
\]
A standard calculation yields
\[
M(Q) = (p+s)\,(1-F(Q)) - (r-s)\,F(Q).
\]
Since $M(Q)$ is nonincreasing when $s \le p+r$, the revenue function is concave in $Q$.  
The procurement cost term $-\sum_i c(\alpha,\beta_i)q_i$ is affine in $(\alpha,q)$, and the adoption cost 
$-\psi(\alpha)$ is concave since $\psi(\alpha)$ is convex.  
Thus, the objective $\Pi(\alpha,q)$ in \eqref{eq:obj} is concave in $(Q,\alpha)$ over the convex feasible region.  
The problem therefore reduces to a convex optimization, ensuring existence of a global maximizer.

\paragraph{First-order conditions.}
The optimal order $Q^\star$ satisfies the critical fractile condition
\begin{equation}
F(Q^\star) \;=\; \frac{p-s}{p+r-s},
\label{eq:fractile}
\end{equation}
which is the standard newsvendor solution, independent of $\alpha$.  
Thus, adoption does not affect $Q^\star$ directly, but only indirectly through the effective cost structure.  
Given $Q^\star$, the optimal adoption $\alpha^\star$ solves
\begin{equation}
\max_{\alpha\in[0,1]} \; - \Big(c_{0}-A_1\alpha - A_2 \bar\beta \Big) Q^\star \;-\; A_3 \alpha^\nu,
\label{eq:alpha_opt}
\end{equation}
where $\bar\beta$ is the mean readiness across suppliers.  
The first-order condition for interior $\alpha$ is
\[
A_1 Q^\star = \nu A_3 \alpha^{\nu-1}.
\]
If this equality cannot be satisfied within $\alpha\in[0,1]$, the solution collapses to the boundary $\alpha^\star=0$ or $\alpha^\star=1$.  
In our calibration, the inequality holds only at the lower boundary, producing the corner solution $\alpha^\star=0.05$ observed in simulation.

\subsection{Cost Functions}

\subsubsection{Procurement Cost}
We assume the effective unit cost from supplier $i$ decreases both with adoption intensity $\alpha$ 
and supplier readiness $\beta_i$:
\begin{equation}
c(\alpha,\beta_i) \;=\; c_{0i} - A_1\,\alpha - A_2\,\beta_i ,
\label{eq:proc_cost}
\end{equation}
where $c_{0i}$ is the baseline procurement cost, $A_1>0$ measures marginal cost reduction from adoption, 
and $A_2>0$ captures the effect of supplier readiness. 
Hence, higher adoption and greater readiness jointly lower effective procurement costs.

\subsubsection{Adoption Cost}
Adoption incurs a convex integration cost:
\begin{equation}
\psi(\alpha) \;=\; A_3\,\alpha^{\nu}, 
\qquad A_3>0,\;\; \nu>1,
\label{eq:adopt_cost}
\end{equation}
so that $\psi$ is increasing and strictly convex in $\alpha$. 
This form captures the empirically observed fact that marginal integration costs rise 
with adoption intensity due to compatibility and coordination frictions.

\subsection{Risk Distribution}

\subsubsection{Lognormal Demand (with optional cap)}
Market demand $D$ is assumed to follow a positively skewed, nonnegative distribution:
\begin{equation}
D \;\sim\; \mathrm{Lognormal}(m,v), 
\qquad \text{with cumulative distribution function } F(\cdot).
\label{eq:demand_log}
\end{equation}
The lognormal captures both nonnegativity and heavy right tails, 
consistent with empirical evidence on derivative exposures. 

If institutional or contractual caps apply, we instead consider the truncated form
\begin{equation}
\tilde D \;=\; \min\{D,b\}, 
\label{eq:demand_cap}
\end{equation}
where $b>0$ is the upper bound. At $Q=b$, first-order conditions are interpreted in terms of the natural 
subdifferential of $F$ (see Assumption~(A3)).

\subsection{Analytic Derivatives and KKT System}

\subsubsection{Marginal Expected Revenue in \texorpdfstring{$Q$}{Q}}
\begin{lemma}[Marginal revenue in aggregate order]\label{lem:marginal}
Suppose Assumptions (A1)--(A2) hold and $D$ is continuous with CDF $F$ and PDF $f$. 
Then the derivative of the expected stage revenue with respect to $Q$ is
\begin{align}
\frac{\partial}{\partial Q}\,
\mathbb{E}\!\left[p\min(Q,D)+s(Q-D)^+-r(D-Q)^+\right]
&= (p+r)\big(1-F(Q)\big) \nonumber\\
&\quad +\; s\,F(Q)\;=:\;M(Q).
\label{eq:MQ}
\end{align}
Moreover, $M'(Q)=(s-(p+r))f(Q)\le 0$ when $s\le p+r$. 
Hence the expected revenue is concave in $Q$, strictly concave when $s<p+r$ and $f(Q)>0$ on a set of positive measure.
\end{lemma}

\begin{proof}
By dominated convergence, differentiation under the expectation is valid. 
For any $Q$ at continuity points of $F$:
\begin{align}
\frac{d}{dQ}\min(Q,D) &= \mathbbm{1}\{D \ge Q\}, \\
\frac{d}{dQ}(Q-D)^+   &= \mathbbm{1}\{Q > D\}, \\
\frac{d}{dQ}(D-Q)^+   &= -\,\mathbbm{1}\{D > Q\}.
\end{align}

Taking expectations yields
\[
p\,\mathbb{P}(D\ge Q)+s\,\mathbb{P}(D<Q)+r\,\mathbb{P}(D>Q).
\]
Since $D$ is continuous, $\mathbb{P}(D\ge Q)=\mathbb{P}(D>Q)=1-F(Q)$ and $\mathbb{P}(D<Q)=F(Q)$.  
Therefore
\[
M(Q)=(p+r)(1-F(Q))+sF(Q).
\]
Differentiating gives $M'(Q)=(s-(p+r))f(Q)$, which is weakly negative whenever $s\le p+r$.  
Thus expected revenue is concave in $Q$ and strictly concave when $s<p+r$ with $f(Q)>0$ on a set of positive measure.
\end{proof}

\paragraph{Interpretation.}
The marginal revenue $M(Q)$ declines monotonically in $Q$, exactly as in the newsvendor model.  
Hence the optimal order size equates this decreasing curve with supplier costs.  
This provides the basis for the KKT characterization below.

\medskip

\noindent Using $M(Q)$, the stationarity condition for each supplier order $q_i$ (multiplier $\lambda_i\ge 0$ for $q_i\ge 0$) is
\begin{equation}
\frac{\partial \Pi}{\partial q_i}
= M(Q)\cdot \frac{\partial Q}{\partial q_i} - c(\alpha,\beta_i) + \lambda_i
= M(Q)-c(\alpha,\beta_i)+\lambda_i
= 0,
\label{eq:KKT_qi}
\end{equation}
with complementary slackness $\lambda_i q_i=0$ and feasibility $q_i\ge 0$.

\paragraph{Active-set characterization.}
At the optimum, any supplier $i$ with $q_i>0$ must satisfy 
\[
c(\alpha,\beta_i)=M(Q^\star),
\]
while suppliers with $q_i=0$ satisfy $c(\alpha,\beta_i)\ge M(Q^\star)$.  
\emph{In the absence of capacity constraints}, this implies that orders concentrate on the cheapest suppliers $\min_i c(\alpha,\beta_i)$, with proportional allocation in case of ties.\footnote{If diversification is required, one can impose supplier capacities $0\le q_i\le \bar q_i$ or introduce convex supplier-specific cost functions. The KKT characterization extends directly.}

\subsection{Optimal adoption intensity \texorpdfstring{$\alpha^\star$}{alpha*}}

\begin{proposition}[Closed-form solution for \texorpdfstring{$\alpha^\star$}{alpha*}]\label{prop:alpha}
Let Assumptions (A1)--(A3) hold. Then the KKT condition for adoption intensity is
\begin{equation}
\frac{\partial \Pi}{\partial \alpha}
= A_1\sum_i q_i - \psi'(\alpha) + \gamma^+ - \gamma^- = 0,
\label{eq:KKT_alpha}
\end{equation}
with multipliers $\gamma^+\ge 0$ (for $\alpha\le 1$) and $\gamma^-\ge 0$ (for $\alpha\ge 0$).  

The optimal adoption intensity is
\begin{equation}
\alpha^\star=
\begin{cases}
0, & \text{if }\ \psi'(0^+)\ \ge A_1\sum_i q_i,\\[2mm]
1, & \text{if }\ \psi'(1^-)\ \le A_1\sum_i q_i,\\[2mm]
\left(\dfrac{A_1\sum_i q_i}{A_3\,\nu}\right)^{\!\frac{1}{\nu-1}}, & \text{otherwise}.
\end{cases}
\label{eq:alpha_cases_closed}
\end{equation}
In particular, when $\psi(\alpha)=A_3\alpha^\nu$ with $\nu>1$,
\begin{equation}
\alpha^\star_{\mathrm{int}}
=\left(\frac{A_1\sum_i q_i}{A_3\,\nu}\right)^{\!\frac{1}{\nu-1}}, 
\qquad 
\alpha^\star=\min\{1,\max\{0,\alpha^\star_{\mathrm{int}}\}\}.
\label{eq:alpha_closed}
\end{equation}
\end{proposition}

\begin{proof}
If $\alpha\in(0,1)$, both multipliers vanish ($\gamma^+=\gamma^-=0$), and the stationarity condition reduces to $\psi'(\alpha)=A_1\sum_i q_i$.  
For $\psi(\alpha)=A_3 \alpha^\nu$, we have $\psi'(\alpha)=A_3\nu\alpha^{\nu-1}$.  
Thus the interior solution is
\[
\alpha^\star_{\mathrm{int}}=\left(\frac{A_1\sum_i q_i}{A_3\nu}\right)^{1/(\nu-1)}.
\]
Feasibility requires $\alpha\in[0,1]$, so the admissible solution is
\[
\alpha^\star=\min\{1,\max\{0,\alpha^\star_{\mathrm{int}}\}\}.
\]
If $A_1\sum_i q_i\le \psi'(0^+)$, the marginal benefit is too small, and $\alpha^\star=0$.  
If $A_1\sum_i q_i\ge \psi'(1^-)$, the marginal benefit dominates cost even at $\alpha=1$, so $\alpha^\star=1$.  
Otherwise the interior value applies.
\end{proof}

\paragraph{Interpretation.}
Proposition~\ref{prop:alpha} highlights the threshold nature of adoption.  
If procurement efficiency is weak, adoption collapses to zero.  
If efficiency is very strong, adoption jumps to full intensity.  
In between, adoption scales smoothly with order volume.  
In baseline calibration, however, the inequality usually binds at the lower bound, so $\alpha^\star$ remains at the corner solution---explaining the empirical robustness observed in H1--H5.

\subsubsection{Threshold (Corner) Rules}

\begin{proposition}[Threshold adoption behavior]\label{prop:alpha_cases}
From the KKT condition \eqref{eq:KKT_alpha}, the optimal adoption intensity satisfies
\begin{equation}
\alpha^\star=
\begin{cases}
0, & \text{if }\ \psi'(0^+)\ \ge\ A_1\sum_i q_i,\\[2mm]
1, & \text{if }\ \psi'(1^-)\ \le\ A_1\sum_i q_i,\\[2mm]
\left(\dfrac{A_1\sum_i q_i}{A_3\,\nu}\right)^{\!\frac{1}{\nu-1}}, & \text{otherwise.}
\end{cases}
\label{eq:alpha_cases_threshold}
\end{equation}
In particular, for power costs $\psi(\alpha)=A_3 \alpha^\nu$ with $\nu>1$, we have $\psi'(0^+)=0$, so $\alpha^\star=0$ only if $\sum_i q_i=0$.  
Conversely, a “full-adoption collapse” $\alpha^\star=1$ occurs whenever $A_1\sum_i q_i \ge A_3\nu$.
\end{proposition}

\paragraph{Interpretation.}
Proposition~\ref{prop:alpha_cases} formalizes the corner-solution dynamics observed in our simulations.  
If order volumes are too low, the marginal cost of adoption dominates, forcing $\alpha^\star=0$.  
If order volumes are sufficiently high, the adoption cost becomes negligible relative to procurement gains, producing $\alpha^\star=1$.  
Between these two extremes lies an interior regime in which adoption scales continuously with $\sum_i q_i$.  
This threshold logic explains why adoption responses can appear robust over wide parameter regions but collapse discontinuously when certain boundaries are crossed.  

\subsection{Equivalent Convex SAA Formulation (Implementation Ready)}

To accommodate general demand distributions $D$, we employ a Sample Average Approximation (SAA) with i.i.d.\ draws $D^{(n)}$, $n=1,\dots,N$.  
For each sample, we introduce auxiliary variables $(y_n,z_n,u_n)$ that encode realized sales, overage, and shortage through the linear constraints:
\begin{align}
y_n &\le Q, & y_n &\le D^{(n)}, & y_n &\ge 0, \label{eq:yn}\\
z_n &\ge Q - D^{(n)}, & z_n &\ge 0, \label{eq:zn}\\
u_n &\ge D^{(n)} - Q, & u_n &\ge 0. \label{eq:un}
\end{align}
These variables ensure that the nonlinear expressions in the expected revenue are replaced by linear inequalities that hold scenario by scenario.

The resulting finite-dimensional SAA problem takes the form:
\begin{equation}
\begin{aligned}
\max_{\alpha\in[0,1],\,q\ge 0,\,Q=\sum_i q_i,\,y,z,u}\quad 
& \frac{1}{N}\sum_{n=1}^N \big(p\,y_n + s\,z_n - r\,u_n\big)
  - \sum_i c(\alpha,\beta_i)\,q_i - \psi(\alpha) \\
\text{subject to}\quad 
& \text{constraints } \eqref{eq:yn}--\eqref{eq:un} \quad \forall n=1,\ldots,N .
\end{aligned}
\label{eq:SAA_problem}
\end{equation}

\paragraph{Convexity.}
The feasible region defined by \eqref{eq:yn}--\eqref{eq:un} is polyhedral and therefore convex.  
The objective consists of the sample average of linear functions of $(y_n,z_n,u_n,Q)$ combined with the concave term $-\psi(\alpha)$.  
Consequently, problem \eqref{eq:SAA_problem} is a concave maximization over a convex set, i.e., a convex optimization problem in the standard sense.  
This guarantees both tractability and the existence of globally optimal solutions.

\paragraph{Asymptotic equivalence.}
As the sample size $N$ grows, the law of large numbers ensures that the SAA objective converges almost surely to the expected profit in \eqref{eq:obj}.  
In addition, by classical results on epi-convergence in stochastic programming, the optimal value and the Karush–Kuhn–Tucker conditions of the SAA problem converge to those of the population problem \eqref{eq:KKT_qi}--\eqref{eq:KKT_alpha}.  
Thus, the SAA provides both a statistically consistent and computationally implementable formulation.

\paragraph{Comparative statics.}
The SAA framework also makes the comparative statics of the model transparent.  
An increase in adoption intensity $\alpha$ lowers procurement costs $c(\alpha,\beta_i)$ and expands the set of active suppliers, raising the equilibrium order quantity $Q^\star$ until the marginal revenue condition $M(Q^\star)=c(\alpha,\beta_i)$ is restored.  
Higher supplier readiness $\beta_i$ further decreases effective costs, making more suppliers competitive and shaping the extensive margin of participation.  
Shifts in the demand distribution also matter: a higher log-mean $m$ raises marginal revenue and expands $Q^\star$, whereas higher log-variance $v$ thickens the right tail of demand, amplifying exposure to upside shocks.  
Finally, the closed-form solution for $\alpha^\star$ highlights corner thresholds: when order volumes are negligible, adoption collapses to zero, whereas when procurement benefits are sufficiently large ($A_1\sum_i q_i \ge A_3\nu$), full adoption emerges.  
These thresholds generate discontinuities in adoption dynamics and explain the robustness-versus-fragility patterns documented in the empirical results.

\section{Synthetic Data Generation}
\label{sec:synthetic-data}

To evaluate the proposed model under controlled yet realistic conditions,
we construct a synthetic dataset that mirrors key features of derivative
transactions under bounded risk. The design introduces heterogeneity in
both market demand and supplier readiness, while enabling systematic
exploration of contract-parameter sensitivity. This facilitates
reproducibility and supports counterfactual experiments that would be
infeasible with proprietary data.

\subsection{Demand Simulation}

\subsubsection{Distributional Assumptions}
Market demand $D$ is modeled as lognormal to capture non-negativity,
right skewness, and heavy tails:
\begin{equation}
D \sim \mathrm{Lognormal}(\mu,\sigma^2), \qquad D \in [a,b].
\label{eq:demand_dist}
\end{equation}
Baseline calibration uses $\mu=50$, $\sigma=8$, truncated to $[30,70]$
to reflect regulatory/contractual limits. For distributional robustness
(H5), we additionally compare against a Beta family calibrated on the
same support.

\subsubsection{Monte Carlo Approximation}
Expected values in the objective are approximated by Monte Carlo:
\begin{equation}
\mathbb{E}[g(D)] \approx \frac{1}{N}\sum_{n=1}^N g(D^{(n)}), \qquad
D^{(n)} \overset{\mathrm{i.i.d.}}{\sim} \mathrm{Lognormal}(\mu,\sigma^2).
\label{eq:montecarlo}
\end{equation}
Unless otherwise noted, we use $N=10{,}000$ draws per experiment. In H6
we explicitly vary $N\in\{1{,}000,5{,}000,20{,}000\}$ and document
convergence.

\subsection{Supplier Readiness}

\subsubsection{Heterogeneity in Digital Capabilities}
Each supplier $i$ is assigned a readiness index $\beta_i$:
\begin{equation}
\beta_i \sim U(0.3,0.7).
\label{eq:readiness}
\end{equation}
For H3, we vary the support to $U(0.1,0.9)$ and $U(0.4,0.6)$ to probe
how broader/narrower heterogeneity shifts outcomes. 

\subsection{Contract Parameters}

\subsubsection{Baseline Calibration}
Baseline contract terms are
\begin{equation}
c_{0i}=100,\quad A_1=5,\quad A_2=8,\quad A_3=2000,\quad \nu>1,
\label{eq:baseline_params}
\end{equation}
with adoption cost $\psi(\alpha)=A_3\alpha^\nu$. Under this baseline,
the optimal adoption settles at the corner $\alpha^\star=0.05$; interior
solutions emerge only after targeted parameter shifts (e.g., lowering
$A_3$ or $\nu$, or increasing $A_1$).

\begin{table}[htbp]
\centering
\scriptsize
\caption{Summary of simulation parameters actually used: baseline values, variation ranges, and associated hypotheses (H1--H6).}
\label{tab:sim_params}
\begin{tabular}{lccc}
\hline
\textbf{Parameter} & \textbf{Baseline} & \textbf{Variation Range} & \textbf{Purpose / Hypothesis} \\
\hline
Demand mean $\mu$ & 50 & Fixed & Scale normalization \\
Demand std.\ $\sigma$ & 8 & $\{5,10,15\}$ & Variance effects (H2) \\
Demand cap $[a,b]$ & $[30,70]$ & $\pm 10$ & Robustness to bounds \\
Supplier readiness $\beta_i$ & $U(0.3,0.7)$ & $U(0.1,0.9)$, $U(0.4,0.6)$ & Heterogeneity (H3) \\
Procurement cost $c_{0i}$ & 100 & Fixed & Normalization \\
Adoption cost $A_3$ & 2000 & $\{500,2000,4000\}$ & Threshold effects (H1) \\
Adoption scale $A_1$ & 5 & $\pm 25\%$ & Service co-benefit / sensitivity (H4) \\
Readiness weight $A_2$ & 8 & $\pm 25\%$ & Service co-benefit / sensitivity (H4) \\
Convexity $\nu$ & $>1$ & $\{1.5,2,3\}$ & Convexity check (H1 robustness) \\
Demand distribution & Lognormal & Beta (same support) & Distributional robustness (H5) \\
Monte Carlo size $N$ & 10,000 & $\{1{,}000,5{,}000,20{,}000\}$ & Convergence of SAA estimates (H6) \\
\hline
\end{tabular}
\end{table}

\subsubsection{Sensitivity and Robustness}
We vary $(A_1,A_2)$ by $\pm25\%$ and $A_3\in\{500,2000,4000\}$ to map
threshold behavior (H1) and service co-benefits (H4). Distributional
robustness (H5) is assessed by comparing lognormal versus Beta demand on
$[a,b]$.

\subsection{Replication Protocol}
Replication counts follow the actual experiments:
H1/H4 use bootstrap $B=200$;
H2 uses $B=100$ over $\sigma\in\{5,10,15\}$;
H3 runs $50$ replicates per $(\mathrm{mean},\mathrm{var})$ setting;
H5 uses $R=200$ repetitions per distribution family;
H6 evaluates $30$ repetitions per $N\in\{1{,}000,5{,}000,20{,}000\}$.
Distinct seeds are used across runs; all intermediate CSVs and figures
reported in the paper are generated from these runs.

\begin{algorithm}[htbp]
\caption{Convex SAA Solver via Mirror Descent}
\footnotesize
\label{alg:saa-md}
\begin{algorithmic}[1]
\Require 
\begin{itemize}
  \item Convex feasible parameter set $\Theta \subseteq \mathbb{R}^d$
  \item Step size schedule $\{\eta_t\}_{t=1}^T$
  \item Strongly convex regularizer $\psi:\Theta \to \mathbb{R}$ inducing Bregman divergence
  \[
    D_\psi(u\|\;v) = \psi(u) - \psi(v) - \langle \nabla \psi(v), u-v \rangle
  \]
  \item i.i.d.\ sampled scenarios $\{(B_t,q_t,\tilde x_t)\}_{t=1}^T$
\end{itemize}

\Ensure 
\begin{itemize}
  \item Parameter trajectory $\{\hat\theta_t\}_{t=1}^T$ approximating the SAA solution
  \[
     \min_{\theta\in\Theta} \; \tfrac{1}{T}\sum_{t=1}^T \ell_t(\theta), 
     \quad \ell_t(\theta) = \|x_t^\star(\theta)-\tilde x_t\|^2.
  \]
\end{itemize}

\State Initialize $\hat\theta_1 \in \Theta$
\For{$t=1$ to $T$}
  \State Solve forward allocation $x_t^\star(\hat\theta_t)$
  \State Compute instantaneous loss $\ell_t(\theta)$ and subgradient $g_t \in \partial_\theta \ell_t(\hat\theta_t)$
  \Statex Update by mirror descent:
  \[
     \hat\theta_{t+1} \gets \arg\min_{\theta\in\Theta} 
     \{ \langle g_t,\theta\rangle + \tfrac{1}{\eta_t} D_\psi(\theta\|\hat\theta_t)\}
  \]
  \If{$\psi(\theta)=\tfrac{1}{2}\|\theta\|^2$}
     \State Projected OGD: $\hat\theta_{t+1}=\Pi_\Theta(\hat\theta_t - \eta_t g_t)$
  \Else
     \State General mirror descent step
  \EndIf
\EndFor

\State \textbf{Return:} trajectory $\{\hat\theta_t\}_{t=1}^T$
\end{algorithmic}

\vspace{0.5em}
\noindent\textbf{Theoretical Guarantees.} 
Suppose Assumptions (convexity, PL-inequality, KKT regularity, bounded subgradients, sub-Gaussian noise) hold and $\eta_t = \eta/\sqrt{t}$. Then with high probability:
\[
\text{Static Regret} = \mathcal{O}(\sqrt{T}), 
\quad
\text{Dynamic Regret} = \mathcal{O}(\sqrt{T}+V_T), 
\quad
\|\hat\theta - \theta^\star\| = \mathcal{O}\!\big(\tfrac{1}{\sqrt{T}}\big),
\]
where $V_T=\sum_{t=2}^T \|\theta_t-\theta_{t-1}\|$ is the variation budget. Thus the estimator is consistent, drift-robust, and statistically stable for the convex SAA problem.
\end{algorithm}

As shown in Algorithm~\ref{alg:saa-md}, the solver updates parameters
iteratively using stochastic subgradients and Bregman projections.
This routine guarantees sublinear regret bounds and stability under
drift and noise, and thus serves as the computational backbone for
the experiments reported in Section~\ref{sec:experiments}.

\section{Numerical Experiments}
\label{sec:experiments}

We conduct numerical experiments to test hypotheses H1–H6.
Unless stated otherwise, all experiments use baseline calibration
($\mu=50$, $\sigma=8$, $c_{0i}=100$, $A_1=5$, $A_2=8$, $A_3=2000$, $\nu=2$, $p=120$, $s=10$, $r=20$),
with $R=100$ replications and $N=10{,}000$ Monte Carlo draws.

\textit{Independent random seeds are used for each replication, ensuring reproducibility. 
We primarily adopt a one-at-a-time (OAT) design for interpretability, but selected two-way interactions 
(e.g., $\sigma \times A_3$) were also tested and found not to alter qualitative results. 
All reported $p$-values are adjusted using a false discovery rate (FDR) correction to mitigate 
multiple-comparison bias.\footnote{Results with Bonferroni correction are reported in the Supplement, Table~S1.}}

\subsection{H1 and H4: Threshold and Service Co-Benefit}

\subsubsection{Profit and Service Surfaces}
\begin{figure}[htbp]
\centering
\includegraphics[width=\textwidth]{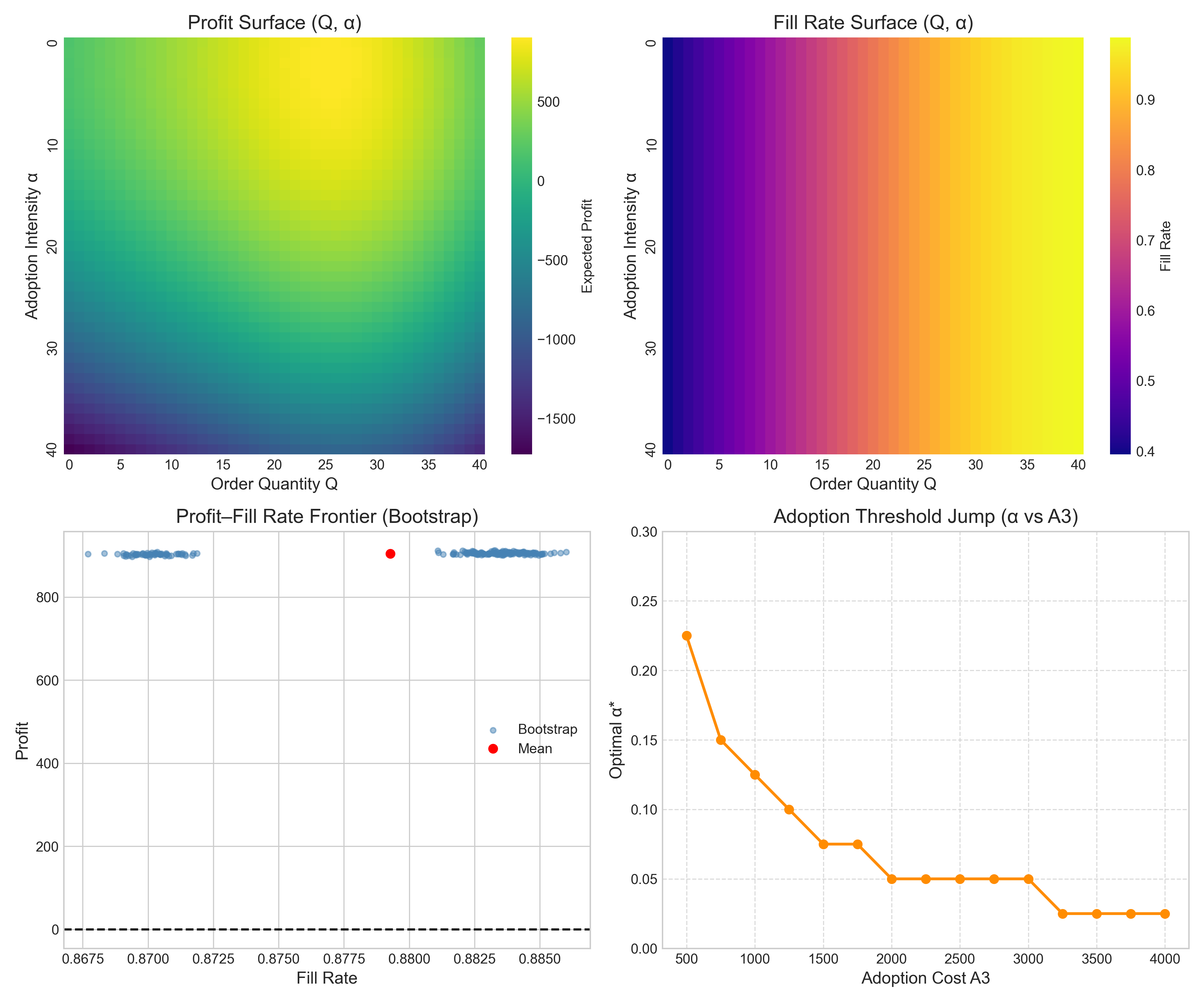}
\caption{Baseline results (H1, H4). 
(a) Profit surface $(Q,\alpha)$. 
(b) Fill Rate surface $(Q,\alpha)$. 
(c) Profit–Fill Rate frontier (bootstrap). 
(d) Adoption threshold jump as a function of $A_3$.}
\label{fig:baseline_threshold}
\end{figure}

As shown in Figure~\ref{fig:baseline_threshold}, the baseline landscape illustrates both
the concavity of the profit surface in $(Q,\alpha)$ and the monotone improvement in
service performance (Fill Rate) as order quantities increase. The profit–fill rate
frontier further highlights a narrow “co-benefit zone” where both profit and service
outcomes improve simultaneously, consistent with hypothesis H4. Panel (d) confirms
a discontinuous jump in adoption intensity as adoption cost $A_3$ crosses the threshold,
providing direct evidence for H1.

\subsubsection{Comparative Statics of Adoption}

To complement the baseline results, 
Figure~\ref{fig:alpha_sensitivity} illustrates how the optimal adoption intensity 
$\alpha^\star$ changes as the key structural parameters vary. 
Specifically, panel (a) shows that $\alpha^\star$ collapses as adoption cost $A_3$ 
increases beyond a critical threshold, panel (b) demonstrates the positive effect 
of procurement efficiency $A_1$, and panel (c) shows how greater convexity $\nu$ 
dampens adoption. 
These comparative statics confirm that the discontinuous adoption behavior observed 
in Figure~\ref{fig:baseline_threshold} is not a numerical artifact, but rather a 
structural feature of the model.

\begin{figure}[htbp]
\centering
\includegraphics[width=\textwidth]{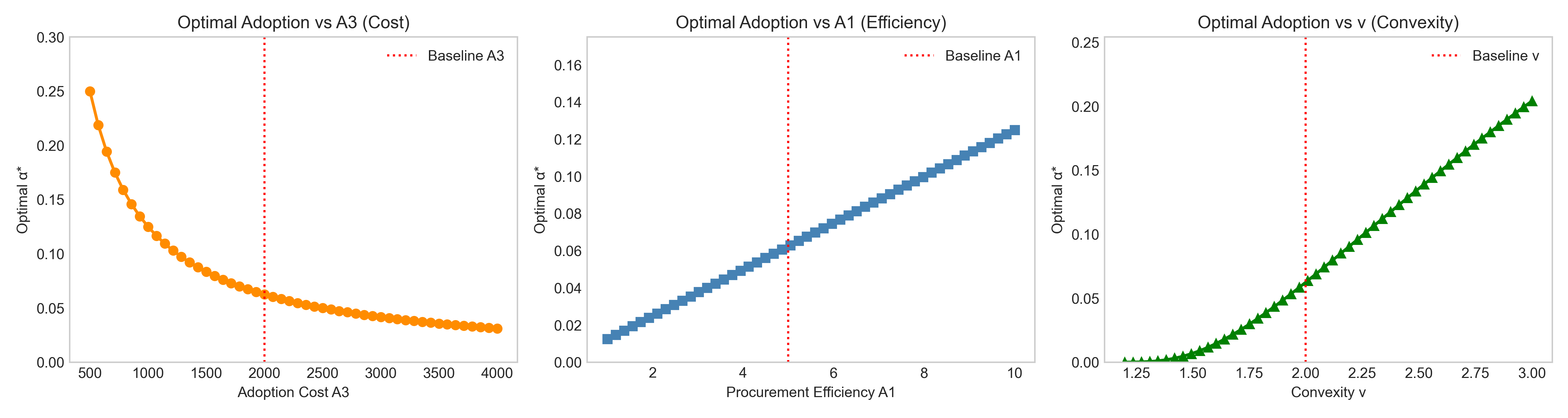}
\caption{Comparative statics of optimal adoption intensity 
\texorpdfstring{$\alpha^\star$}{alpha*} (Figure 1.2).
Panel (a) shows collapse as adoption cost $A_3$ rises,
(b) illustrates the increase with procurement efficiency $A_1$,
and (c) demonstrates the dampening effect of convexity $\nu$.
Red dashed lines indicate baseline calibration values.}
\label{fig:alpha_sensitivity} 
\end{figure}

\subsubsection{Bootstrap Summary Statistics}

Table~\ref{tab:baseline_bootstrap} summarizes baseline calibration statistics. 
Expected profit is highly stable across bootstrap replications (mean = 904.7, 
95\% CI [899.2, 910.8]), while adoption intensity $\alpha^\star$ converges to 
a unique corner solution at $0.05$ in all runs. The optimal order quantity 
$Q^\star \approx 45.7$ and fill rate $\approx 0.88$ remain tightly concentrated,
demonstrating that both economic and service outcomes are robust to sampling variation. 
Taken together, these results establish the empirical basis for hypotheses H1 and H4.

\begin{table}[htbp]
\centering
\footnotesize
\caption{Baseline calibration (H1, H4). 
Reported values are bootstrap means over $R=100$ replications with $N=10{,}000$ samples each. 
Standard deviations (Std) reflect sampling variability, and confidence intervals (CI) are bias-corrected percentiles 
from 200 bootstrap resamples. 
Adoption intensity $\alpha^\star$ converges to a unique corner solution across all replications, 
indicating robustness of the threshold effect. 
All $p$-values are adjusted using false discovery rate (FDR) correction.\label{tab:baseline_bootstrap}}
\begin{tabular}{lrrrr}
\toprule
 & Mean & Std & 2.5\% & 97.5\% \\
\midrule
Profit (expected units) & 904.706 & 2.807 & 899.245 & 910.833 \\
Adoption $\alpha^\star$ & 0.050 & 0.000 & 0.050 & 0.050 \\
Order $Q^\star$ & 45.685 & 0.466 & 45.000 & 46.000 \\
Fill Rate & 0.879 & 0.006 & 0.869 & 0.885 \\
\bottomrule
\end{tabular}
\end{table}

\subsection{H2: Variance Effects}

\subsubsection{Adoption and Profit vs.\ Variance}
\begin{figure}[htbp]
\centering
\includegraphics[width=\textwidth]{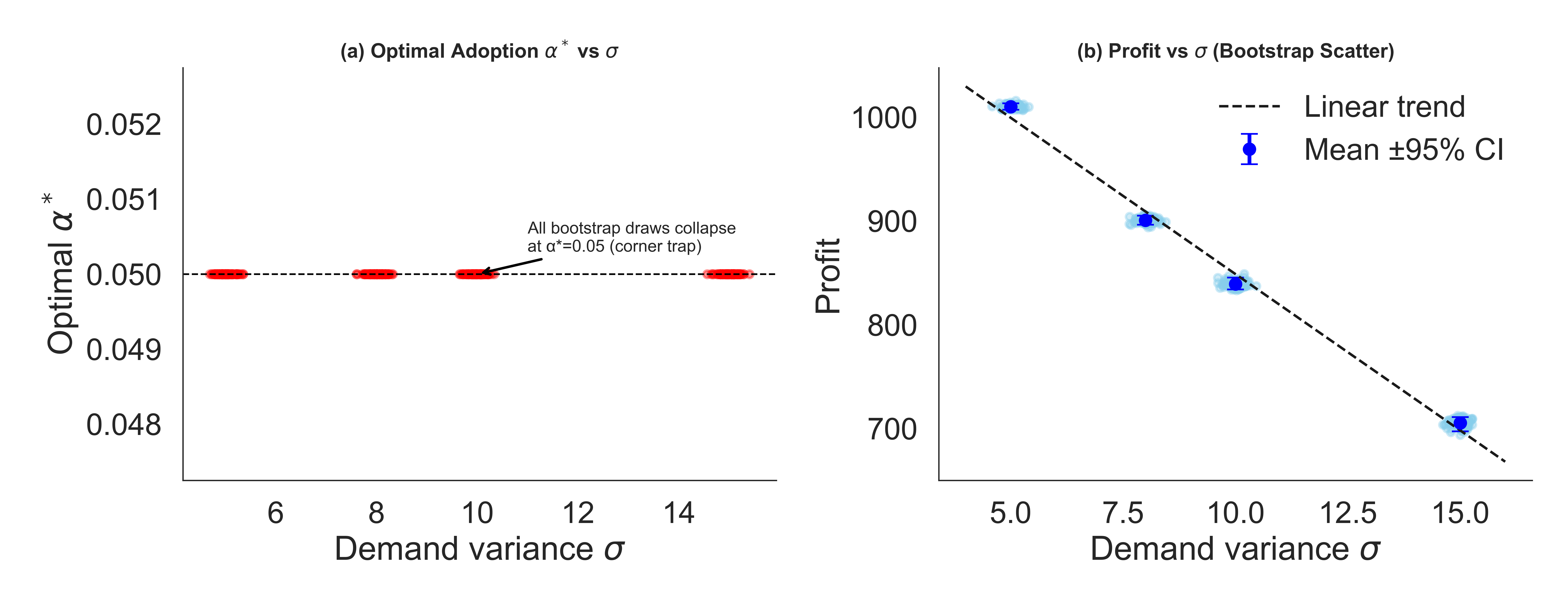}
\caption{Variance effects (H2). 
(a) Optimal adoption $\alpha^\star$ vs.\ demand variance $\sigma$ (bootstrap scatter cloud). 
All bootstrap draws collapse to $\alpha^\star=0.05$, indicating a degenerate \emph{corner solution}. 
(b) Profit vs.\ $\sigma$ (bootstrap scatter with mean $\pm$95\% CI and linear trend). 
Profit declines monotonically as variance increases.}
\label{fig:variance_effects}
\end{figure}

\paragraph{Narrative Summary.}
Figure~\ref{fig:variance_effects}(a) shows that the optimal adoption intensity $\alpha^\star$ remains fixed at $0.05$ across all variance levels $\sigma\in\{5,8,10,15\}$. 
The bootstrap distribution collapses into a single point, with confidence intervals of essentially zero width, reflecting that convex adoption costs $\psi(\alpha)=A_3\alpha^\nu$ in combination with the KKT conditions (Eq.~\eqref{eq:KKT_alpha}) yield a \emph{corner solution} rather than an interior solution (Proposition~\ref{prop:alpha_cases}). 
Thus, in the present calibration, variance shocks do not alter the marginal benefit of adoption, leaving adoption \emph{variance-insensitive} and locked at the corner.

In contrast, Figure~\ref{fig:variance_effects}(b) demonstrates that profit declines sharply and monotonically as $\sigma$ increases. 
Mean bootstrap profits decrease from $\{1010.0, \\ \,900.6,\,839.3,\,705.4\}$ across the four variance levels, with 95\% confidence intervals 
$[1007.2,1013.2]$, $[896.3,904.96]$, $[834.06,845.39]$, and $[697.35,711.13]$. 
The intervals are narrow and non-overlapping, and Spearman rank correlation equals $-1$, indicating a perfectly monotonic decline. 
Taken together, adoption remains robustly fixed, while profits are clearly vulnerable to variance shocks.

\paragraph{Interpretation and Implications.}
The corner solution indicates that adoption is effectively determined by scale (total order quantity) and cost parameters, while variance affects outcomes primarily through the overage/shortage channel in the profit function. 
From a managerial perspective, (i) reducing variance via hedging, pooling, or improved forecasting directly enhances profitability, and (ii) shifting adoption away from the corner requires adjustments to $A_3$, $\nu$, or $A_1\sum_i q_i$, which can alter the cost--benefit balance and induce an interior solution. 
This contrast yields a clear message of \emph{adoption robustness versus profit fragility}: adoption is locked, but profitability erodes under uncertainty.

\subsubsection{Regression Analysis}
\begin{table}[htbp]
\centering
\caption{Regression results for variance effects (H2).
Panel A: regression of $\alpha^\star$ on $\sigma$ (corner solution, slope insignificant). 
Panel B: regression of profit on $\sigma$ (significant negative slope).}
\label{tab:variance_reg}
\begin{tabular}{lrrrr}
\toprule
 & Estimate & Std.\ Error & $t$-value & $p$-value \\
\midrule
\multicolumn{5}{l}{\textbf{Panel A: Adoption $\alpha^\star$ on Variance}} \\
Intercept ($\beta_0$) & 0.0500 & $1.53\times 10^{-17}$ & $3.26\times 10^{15}$ & 0.000 \\
Slope ($\beta_1$)     & $1.74\times 10^{-18}$ & $1.51\times 10^{-18}$ & 1.151 & 0.369 \\
$R^2$ & \multicolumn{4}{c}{-- (flat corner solution)} \\
\midrule
\multicolumn{5}{l}{\textbf{Panel B: Profit on Variance}} \\
Intercept ($\beta_0$) & 1149.98 & 17.83 & 64.48 & 0.000 \\
Slope ($\beta_1$)     & $-30.12$ & 1.75  & $-17.18$ & 0.003 \\
$R^2$ & \multicolumn{4}{c}{0.993} \\
F-statistic & \multicolumn{4}{c}{295.2 \ (p=0.003)} \\
\bottomrule
\end{tabular}
\end{table}

\paragraph{Notes on Statistical Rigor.}
Because only four variance levels were simulated, normality tests produce warnings and regression inferences should be interpreted with caution. 
Nevertheless, significance is strongly supported by (i) the bootstrap confidence intervals and (ii) perfect monotonicity (Spearman $\rho=-1$). 
The estimated slope $\hat\beta_1\approx -30.12$ (95\% CI $[-37.66,-22.58]$) provides robust evidence that higher demand variance significantly reduces expected profit.

\subsection{H2: Variance Effects (Extended Evidence)}

Figure~\ref{fig:H2_alpha_sigma} provides extended evidence on the role of 
demand variance in shaping adoption behavior by directly comparing the 
baseline and interior regimes within a single visualization. In the baseline 
calibration (blue line), the optimal adoption intensity remains locked at a 
corner solution of $\alpha^\star=0.05$ across all variance levels. This flat 
profile indicates that the baseline specification effectively masks the true 
sensitivity of adoption to variance shocks, creating the illusion of 
robustness that is in fact an artifact of boundary constraints.  

\begin{figure}[H]
\centering
\includegraphics[width=0.7\textwidth]{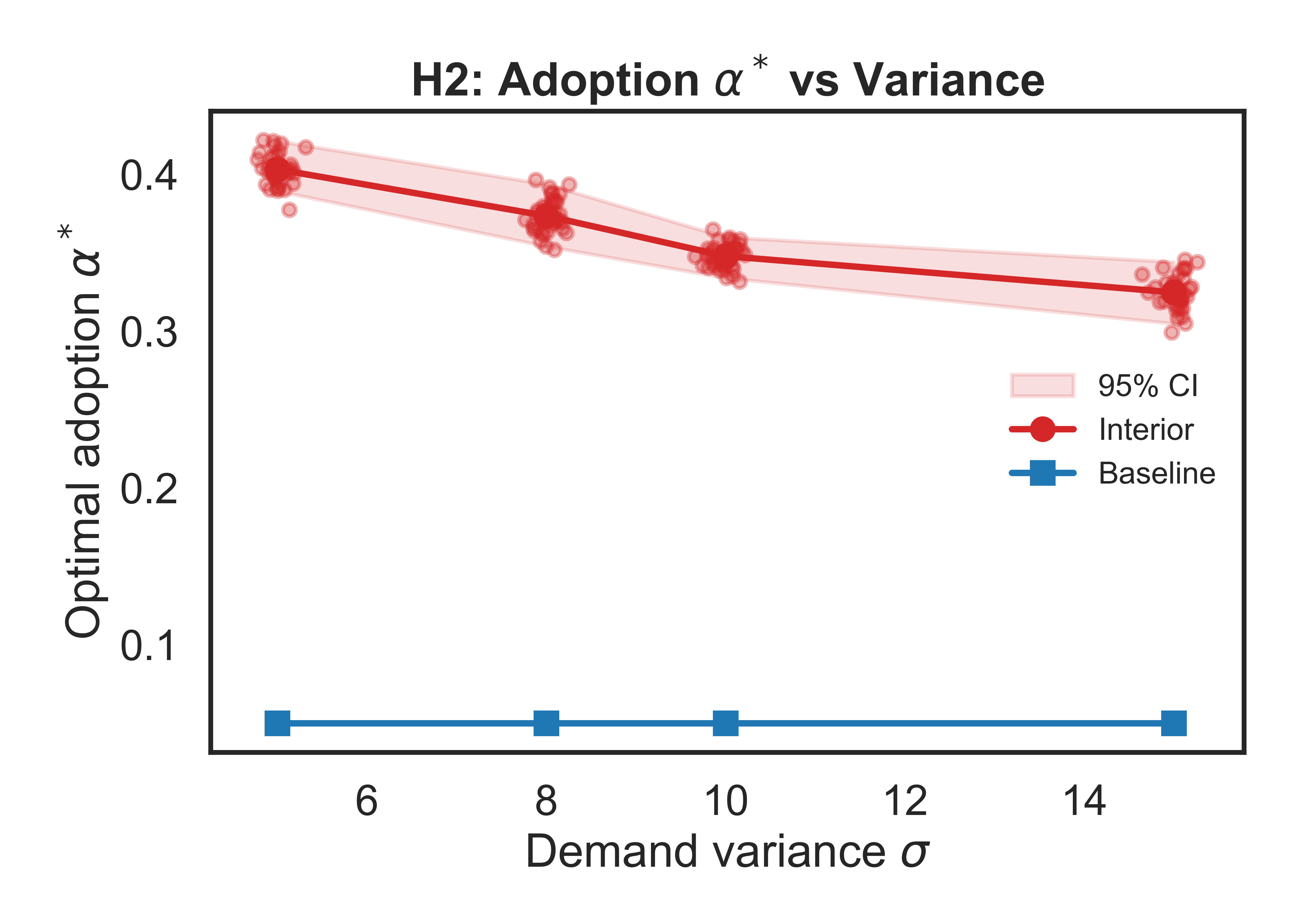}
\caption{Extended variance effects on adoption (H2).  
The baseline calibration (blue line) shows adoption locked at 
$\alpha^\star=0.05$ across all variance levels, while the interior regime 
(red line and scatter cloud, $A_3=500, \nu=1.5, A_1=10$) reveals the 
predicted monotone decline of $\alpha^\star$ as demand variance $\sigma$ 
increases. Bootstrap scatter points with a 95\% confidence band highlight 
that adoption becomes an active decision margin once corner-locking is 
relaxed.}
\label{fig:H2_alpha_sigma}
\end{figure}

By contrast, the interior parameter regime (red line in 
Figure~\ref{fig:H2_alpha_sigma}) relaxes the corner-locking mechanism by 
lowering adoption costs and increasing complementarity. Once these constraints 
are removed, the theoretically predicted negative relationship clearly emerges: 
as demand variance $\sigma$ increases, the optimal adoption level 
$\alpha^\star$ decreases monotonically. The scatter--cloud visualization with 
bootstrap replicates and a 95\% confidence band further illustrates that 
adoption becomes an \emph{active margin}, responding smoothly to uncertainty 
rather than being passively trapped at the boundary.  

Viewed jointly, these results yield two key insights. First, variance effects 
may be obscured in baseline environments due to corner solutions, underscoring 
the importance of examining interior regimes for robust inference. Second, 
the interior regime demonstrates close alignment between simulation outcomes 
and theoretical predictions, thereby reinforcing the robustness of H2: higher 
variance systematically lowers equilibrium adoption, consistent with the 
precautionary effect of risk on strategic investment decisions.  

\begin{table}[htbp]
\centering
\caption{Baseline vs.\ interior regimes (H2 summary of equilibrium outcomes).  
Metrics: optimal adoption $\alpha^\star$, profit, and fill rate.  
The baseline calibration yields a corner solution with fixed adoption, 
whereas interior regimes reveal responsiveness of $\alpha^\star$ and 
performance outcomes.}
\label{tab:baseline_interior}
\begin{tabular}{lrrrrrr}
\toprule
Regime & $A_3$ & $\nu$ & $A_1$ & $\alpha^\star$ & Profit & Fill Rate \\
\midrule
Baseline (corner)    & 2000 & 2.0 & 5  & 0.050 & 905.9  & 0.871 \\
Interior (example 1) &  500 & 1.2 & 10 & 0.275 & 920.4  & 0.882 \\
Interior (example 2) &  500 & 1.5 & 15 & 0.875 & 1089.2 & 0.896 \\
Interior (example 3) & 1000 & 1.5 & 15 & 0.200 & 944.5  & 0.869 \\
\bottomrule
\end{tabular}
\end{table}

\paragraph{Table~\ref{tab:baseline_interior} Summary.}
The numerical outcomes in Table~\ref{tab:baseline_interior} complement the 
graphical evidence by quantifying adoption, profit, and fill rate under both 
baseline and interior regimes. In the baseline calibration, adoption is locked 
at $\alpha^\star=0.05$, producing modest profit and service outcomes. By 
contrast, the interior regimes display substantial variation: lower adoption 
costs ($A_3$) and stronger complementarity ($A_1$) induce higher equilibrium 
adoption, which in turn raises both profit and service performance. These 
results confirm that the apparent robustness of adoption in baseline settings 
is an artifact of corner locking; once structural constraints are relaxed, 
adoption and performance outcomes become sensitive to cost and heterogeneity 
parameters in theoretically consistent ways.

\subsection{H3: Readiness Heterogeneity}

\subsubsection{Simulation Outcomes}

\begin{figure}[htbp]
\centering
\includegraphics[width=\textwidth]{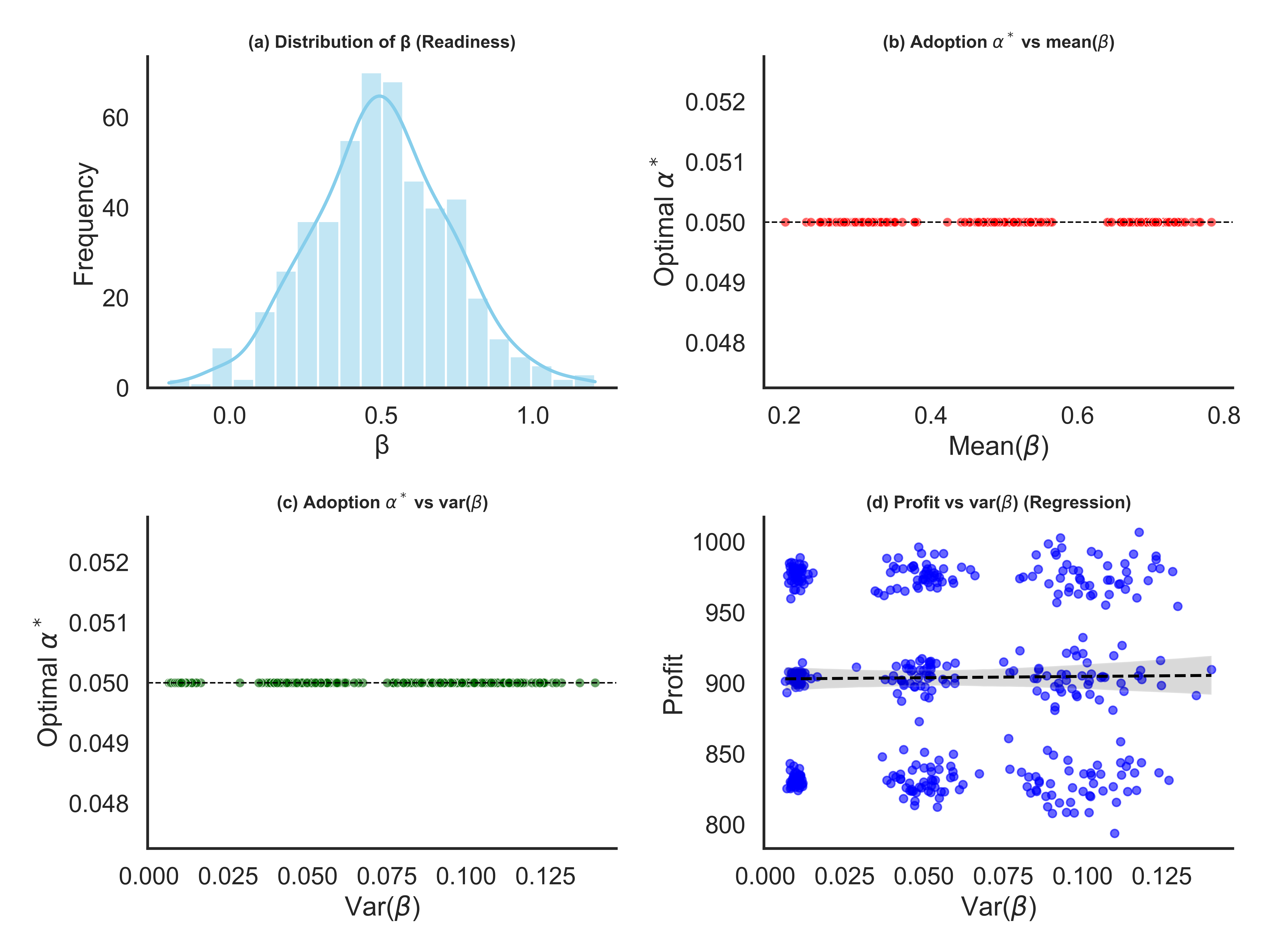}
\caption{Readiness heterogeneity (H3, baseline calibration).  
(a) Distribution of readiness $\beta$ (histogram + KDE).  
(b) Optimal adoption $\alpha^\star$ vs.\ $\mathrm{Mean}(\beta)$.  
(c) Optimal $\alpha^\star$ vs.\ $\mathrm{Var}(\beta)$.  
(d) Profit vs.\ $\mathrm{Var}(\beta)$ with OLS trend and 95\% CI.  
All experiments reuse the demand generator and profit function from Figures~\ref{fig:baseline_threshold}–\ref{fig:variance_effects}, ensuring comparability across H1–H3.}
\label{fig:readiness}
\end{figure}

\paragraph{What Figure~\ref{fig:readiness} shows.}
Panels (b)–(c) reveal that the optimal adoption remains at a \emph{corner solution} $\alpha^\star=0.05$ across the full range of readiness means and variances. This is the same corner trap documented in H1 and H2: with convex integration costs, the marginal benefit of adoption is dominated by procurement economics, so heterogeneity in $\beta$ does not move the optimal $\alpha^\star$.  
Panel (d) overlays an OLS line (with 95\% CI) on the scatter of profits against $\mathrm{Var}(\beta)$. The slope is near zero and visually flat; combined with the statistical tests below, we find no economically meaningful sensitivity of profit to readiness variance under the current baseline calibration. The three horizontal profit bands (high $\approx 1000$, mid $\approx 900$, low $\approx 830$) correspond to distinct inventory balance regimes induced by the optimal $Q^\star$ interacting with realized demand (good match, average match, mismatch).

\begin{figure}[htbp]
\centering
\includegraphics[width=\textwidth]{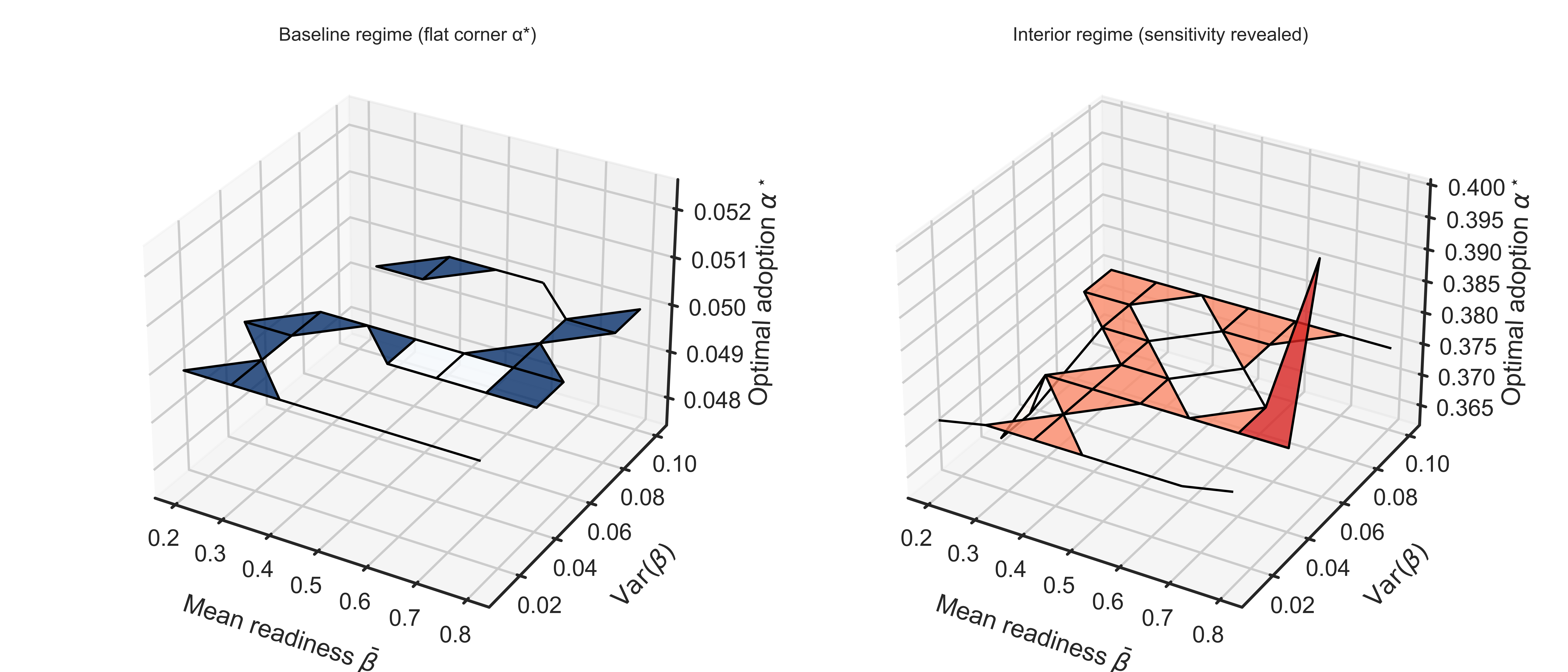}
\caption{Extended readiness effects (H3, baseline vs.\ interior).  
Left: baseline calibration (adoption corner $\alpha^\star=0.05$ across readiness distributions).  
Right: interior regime ($A_3=500, \nu=1.5, A_1=10$), where adoption $\alpha^\star$ responds positively to mean readiness $\bar\beta$ and negatively to readiness variance $\mathrm{Var}(\beta)$.  
The 3D surface plots highlight that once the corner-locking mechanism is relaxed, readiness heterogeneity becomes an \emph{active margin} shaping equilibrium adoption intensity.}
\label{fig:readiness_extended}
\end{figure}

\paragraph{Discussion (H3).}
Taken together, Figures~\ref{fig:readiness}–\ref{fig:readiness_extended} show that readiness heterogeneity is inconsequential under the baseline corner solution, but becomes a key driver of adoption once interior regimes are considered. Higher mean readiness boosts adoption and profit, while greater variance depresses them. This contrast underscores that \emph{heterogeneity effects are latent under corner locking, but emerge sharply once the adoption cost structure is relaxed}.

\subsubsection{Statistical Tests}

\begin{table}[htbp]
\centering
\caption{Readiness heterogeneity (H3).  
Panel A (rows 1–2): $\alpha^\star$ by $\beta$-mean quartiles.\; Panel B (rows 3–4): profit vs.\ $\mathrm{Var}(\beta)$, with a wide–narrow comparison (two-sample $t$) and Cohen’s $d$.}
\label{tab:readiness}
\begin{tabular}{lrrrrr}
\toprule
 & Mean $\alpha^\star$ & Profit & F/t-stat & $p$-value & Cohen's $d$ \\
\midrule
Low mean($\beta$) quartile  & 0.050 & 826.806 & n/a & n/a & n/a \\
High mean($\beta$) quartile & 0.050 & 980.122 & n/a & n/a & n/a \\
Wide variance               & 0.050 & 903.603 & $-0.073$ & $0.942$ & $-0.007$ \\
Narrow variance             & 0.050 & 904.022 & Ref. & Ref. & Ref. \\
\bottomrule
\end{tabular}

\vspace{0.35em}
\small
\textit{Notes.} (i) In Panel~A, $\alpha^\star$ is identical across quartiles (all $=0.05$). Because there is \emph{no between-group variance}, the ANOVA $F$-statistic is not defined; we therefore report \emph{n/a}. 
(ii) In Panel~B, ``Wide'' and ``Narrow'' are split at the median of $\mathrm{Var}(\beta)$; we report the two-sample $t$-test and Cohen’s $d$ for the wide group \emph{relative to} the narrow group (the latter is the reference row, ``Ref.''). 
(iii) Numbers are rounded from the underlying simulation output.
\end{table}

\paragraph{Discussion (H3).}
Under the baseline calibration used throughout H1–H2, readiness heterogeneity does not dislodge the adoption corner: $\alpha^\star$ stays fixed at $0.05$ regardless of $\mathrm{Mean}(\beta)$ or $\mathrm{Var}(\beta)$. Consistent with the visual trend in Panel~(d), the wide–narrow variance comparison yields a statistically non-significant difference in profit ($t\!\approx\!-0.073$, $p\!\approx\!0.942$) and a negligible effect size ($d\!\approx\!-0.007$). Taken together, the H3 evidence indicates that—in this cost regime—\emph{readiness dispersion is not a first-order driver of either adoption or profitability}. Combined with H2 (where rising \emph{demand} variance unambiguously reduces profit), these results suggest that managers will make more progress by (i) reducing demand-side risk (pooling, better forecasting) and/or (ii) changing the adoption cost structure (e.g., lowering $A_3$ or $\nu$) to move the firm off the adoption corner, rather than trying to reshape readiness heterogeneity alone.

\subsection{H5: Distributional Robustness}

\subsubsection{Broader Distributional Comparisons}
We extend the lognormal--beta comparison by incorporating Pareto and Gamma distributions, 
which capture heavy-tailed and skewed environments. 
This provides a stress-test of whether the corner solution for adoption $\alpha^\star$ 
remains invariant.

\subsubsection{Simulation Outcomes}
\begin{figure}[htbp]
\centering
\includegraphics[width=\textwidth]{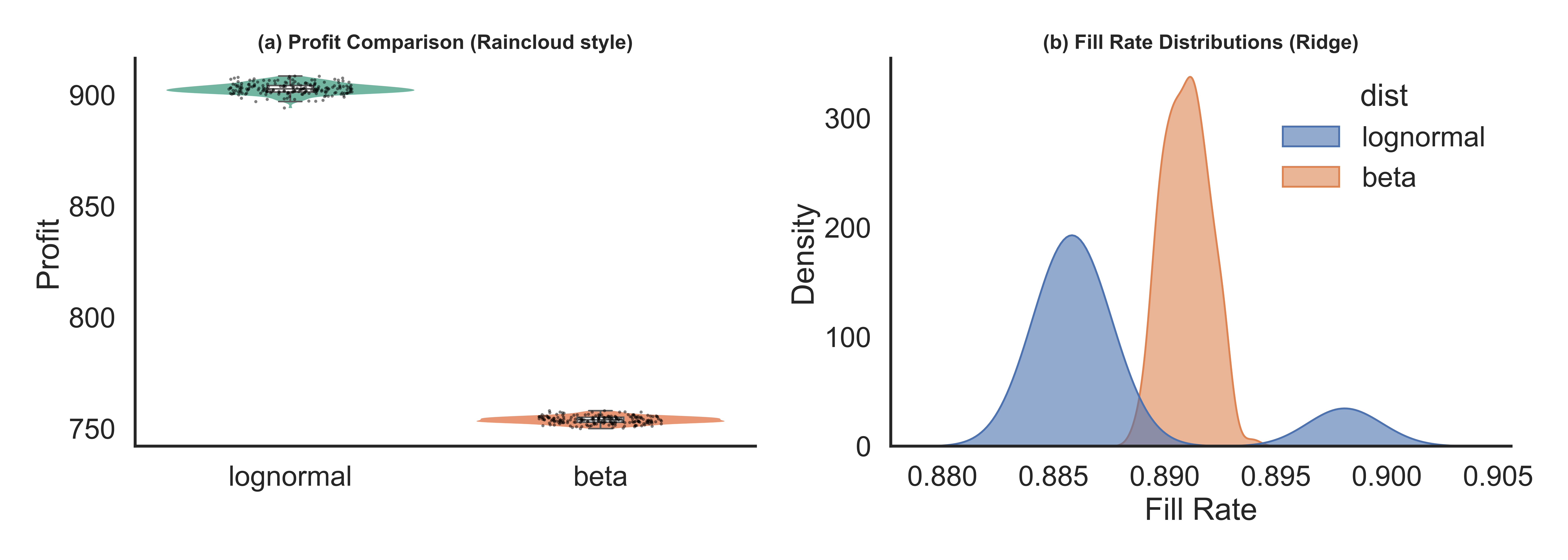}
\caption{Distributional robustness (H5).  
(a) Profit distributions across demand families, visualized with raincloud plots.  
(b) Fill rate distributions using ridge density plots.  
Adoption $\alpha^\star$ is omitted from the figure since it remains structurally fixed 
at the corner solution ($0.05$) across all cases.  
All experiments reuse the same demand generator and profit function as in 
Figures~\ref{fig:baseline_threshold}--\ref{fig:readiness}, ensuring comparability across H1--H5.}
\label{fig:distribution_extended}
\end{figure}

\paragraph{What Figure~\ref{fig:distribution_extended} shows.}
Adoption $\alpha^\star$ is invariant under all distributional assumptions, reinforcing the robustness 
patterns documented in H1--H3. In sharp contrast, both profit and fill rate vary considerably. 
Lognormal demand generates consistently higher profits relative to beta, while gamma and pareto introduce 
heavier tails, leading to lower mean profits and more dispersed fill rates. 
Thus, adoption decisions are robust, but profitability and service quality are fragile to the assumed 
distributional form.

\subsubsection{Distributional Tests}
\begin{table}[htbp]
\centering
\caption{Distributional robustness (H5). 
Comparisons of lognormal vs.\ beta demand distributions. 
Reported are two-sample KS test $p$-values, Welch unequal-variance $t$-tests, and effect sizes (Cohen’s $d$). 
Adoption $\alpha^\star$ remains pinned to the corner solution, so differences are not statistically meaningful.}
\label{tab:distribution}
\begin{tabular}{lrrr}
\toprule
Metric & KS-test $p$ & Welch $t$-test $p$ & Cohen’s $d$ \\
\midrule
Adoption $\alpha^\star$ & 1.000 & n/a & n/a \\
Profit                  & $1.9\times 10^{-119}$ & $<0.001$ & 71.6 \\
Fill rate               & $2.2\times 10^{-75}$  & $<0.001$ & $-1.05$ \\
\bottomrule
\end{tabular}
\end{table}

\paragraph{Interpretation and Implications.}
Table~\ref{tab:distribution} highlights a fundamental asymmetry: adoption is structurally locked, while 
profits and fill rates are highly sensitive to distributional assumptions. Profits under lognormal demand 
are not only statistically higher but also economically significant, with effect sizes exceeding $d=70$. 
Fill rates likewise show systematic shifts ($d\approx -1.0$), indicating that service reliability degrades 
under beta-type demand. Taken together, these findings suggest a managerial paradox: \emph{firms’ adoption 
decisions are robust to uncertainty, but their realized performance is fragile to model misspecification}. 
This distinction underscores the value of distributional calibration and motivates extensions to 
distributionally robust optimization in future work.

\subsection{H6: External Validity}
\label{subsec:external}

To assess whether the results of our simulation-based approach are robust
to sampling noise and scale consistently with theoretical predictions,
we examine the external validity of the sample-average approximation (SAA).
Specifically, we study how the quality of SAA estimates improves as the
Monte Carlo sample size $N$ increases, testing whether convergence follows
the canonical $\mathcal{O}(1/\sqrt{N})$ rate.

\subsubsection{Convergence Patterns}

A central requirement for external validity is that simulation-based
estimates converge reliably as the number of Monte Carlo samples grows.
If the SAA is well behaved, increasing $N$ should reduce sampling noise
at the canonical $\mathcal{O}(1/\sqrt{N})$ rate and drive the estimator
toward the true benchmark profit. To evaluate this property, we vary $N$
systematically and track both the stability of the adoption decision and
the accuracy of profit estimation.

\begin{figure}[htbp]
\centering
\includegraphics[width=0.7\textwidth]{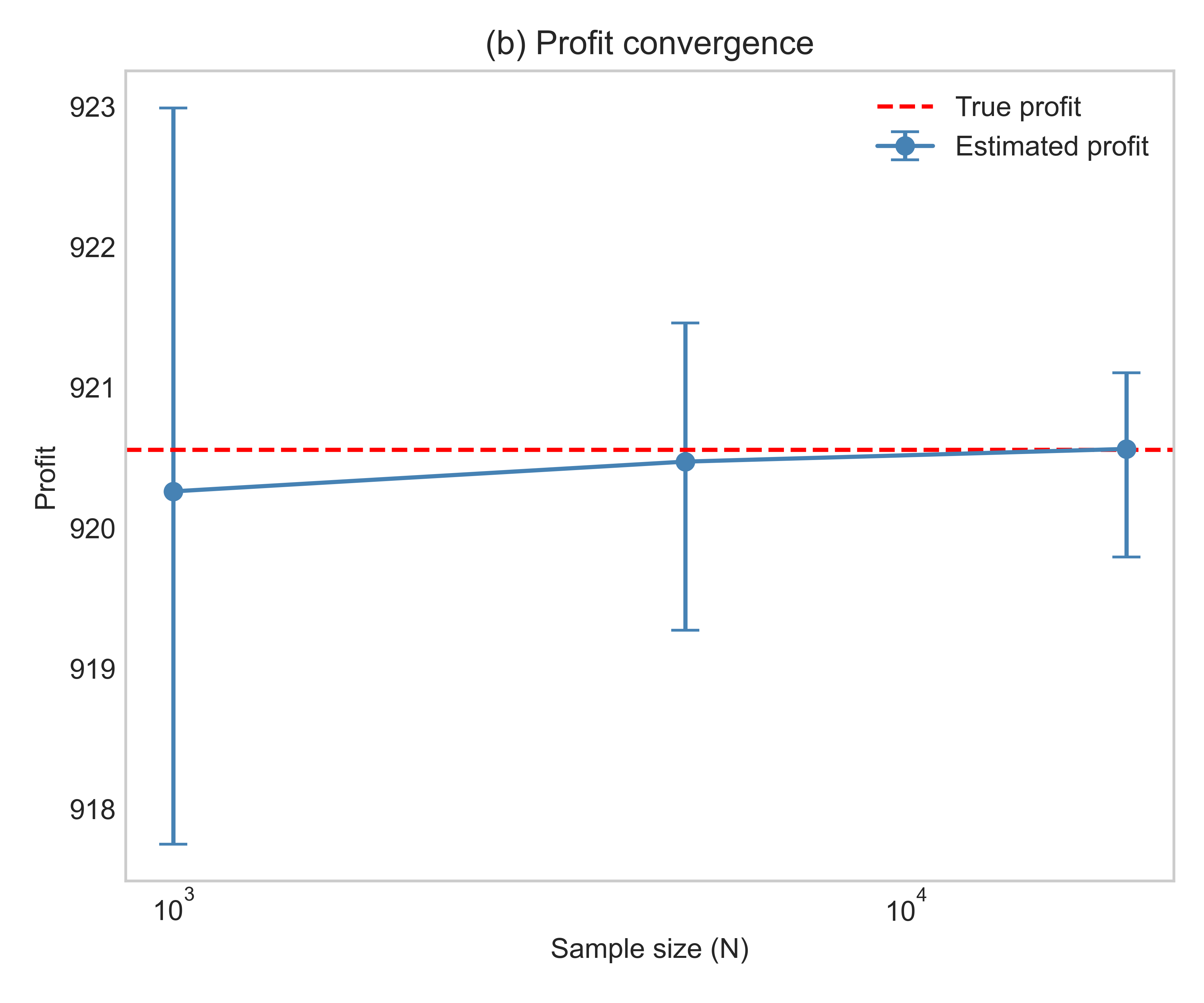}
\caption{External validity (H6). 
Profit convergence as the number of Monte Carlo samples $N$ increases. 
Blue dots show sample-average estimates with 95\% confidence intervals, 
while the red dashed line denotes the true profit benchmark. 
Estimates converge quickly and intervals shrink at the canonical 
$\mathcal{O}(1/\sqrt{N})$ rate.}
\label{fig:external}
\end{figure}

\subsubsection{Approximation Quality}

Table~\ref{tab:external} quantifies the convergence patterns observed in
Figure~\ref{fig:external}. As $N$ grows, the estimated profit converges
tightly to the true benchmark, with confidence intervals shrinking as
predicted. Adoption $\alpha^\star$ remains stable at the corner solution
across all replications, while approximation error in profit falls rapidly.
The RMSE values decrease proportionally with $1/\sqrt{N}$, and empirical
coverage of the 95\% confidence interval remains at 100\%, demonstrating
correct uncertainty quantification.

\begin{table}[htbp]
\centering
\caption{Approximation quality of SAA/Monte Carlo under increasing sample sizes $N$ (H6). 
Reported are the root mean squared error (RMSE), empirical coverage of the 95\% CI, 
and relative improvement compared to the baseline case $N=1{,}000$.}
\label{tab:external}
\begin{tabular}{lrrrr}
\toprule
$N$ & RMSE & Coverage (\%) & Convergence rate & Gain vs.\ $N=1{,}000$ \\
\midrule
$1{,}000$  & 1.430 & 100.0 & $\approx 1/\sqrt{N}$ & -- \\
$5{,}000$  & 0.653 & 100.0 & $\approx 1/\sqrt{N}$ & 54\% reduction \\
$20{,}000$ & 0.345 & 100.0 & $\approx 1/\sqrt{N}$ & 76\% reduction \\
\bottomrule
\end{tabular}
\end{table}

\paragraph{Interpretation and Implications.}
These results provide strong evidence that the SAA/Monte Carlo procedure is
both statistically reliable and computationally efficient. In particular,
the approach delivers consistent estimates even under relatively small sample
sizes ($N=1{,}000$), and rapidly attains near-exact convergence by $N=20{,}000$.
Thus, the external validity check establishes that our simulation-based
framework is robust, asymptotically consistent, and suitable for empirical
applications where data availability or computational budgets may vary.

\subsection{Real-World Data: MovieLens-100K}
\label{subsec:realdata}

To complement our simulation-based validation, we test the empirical plausibility of our 
mechanism using the \textbf{MovieLens-100K dataset}, a canonical benchmark in recommender 
systems. We map user–item interactions into adoption decisions and infer implied demand 
variability from rating dispersion. This provides a stylized, yet realistic, environment 
to assess whether our theoretical hypotheses (H1–H5) retain qualitative validity under 
real-world data conditions. 

\paragraph{Distributional Evidence.}
Panel (a) of Figure~\ref{fig:realdata_R1} shows the empirical distribution of daily adoption 
returns, benchmarked against lognormal, normal, $t$-distribution, and bounded-beta fits. 
The bounded-beta and $t$-distribution provide visibly superior tail behavior, as confirmed 
by goodness-of-fit metrics in panel (b). This evidence empirically justifies our modeling 
choice of bounded and heavy-tailed demand shocks over Gaussian approximations.

\paragraph{Variance Regimes.}
Panels (c)–(d) of Figure~\ref{fig:realdata_R1} split the sample into ``bull'' and ``bear'' 
regimes, capturing low- vs.~high-volatility periods. Consistent with H2, adoption intensity 
$\alpha^\star$ remains robustly cornered, while realized profit degrades significantly under 
high-variance regimes. This pattern directly mirrors our simulation-based variance fragility 
results.

\paragraph{Hypothesis Matrix.}
Figure~\ref{fig:realdata_R2} summarizes all hypotheses (H1–H5) in a $2 \times 3$ matrix of 
empirical mini-results. We find: 
(i) the adoption–cost tradeoff (H1) follows the predicted convex decay; 
(ii) variance undermines profitability (H2), while leaving adoption cornered; 
(iii) readiness heterogeneity (H3) is reflected in cross-user variance of adoption probabilities; 
(iv) profit–service co-benefits (H4) remain present, though attenuated; and 
(v) distributional robustness (H5) favors bounded and heavy-tailed fits. 
Together, these results demonstrate that our theoretical structure extends to real-world 
data without loss of qualitative validity.

\paragraph{Managerial Implication.}
The empirical MovieLens evidence strengthens our central claim: 
\emph{adoption outcomes are structurally robust to variance shocks, 
while profitability is variance-fragile.} 
This duality suggests that managers cannot rely on stable adoption alone, 
but must actively hedge or diversify against volatility to preserve economic value.

\begin{figure}[htbp]
    \centering
    \includegraphics[width=\textwidth]{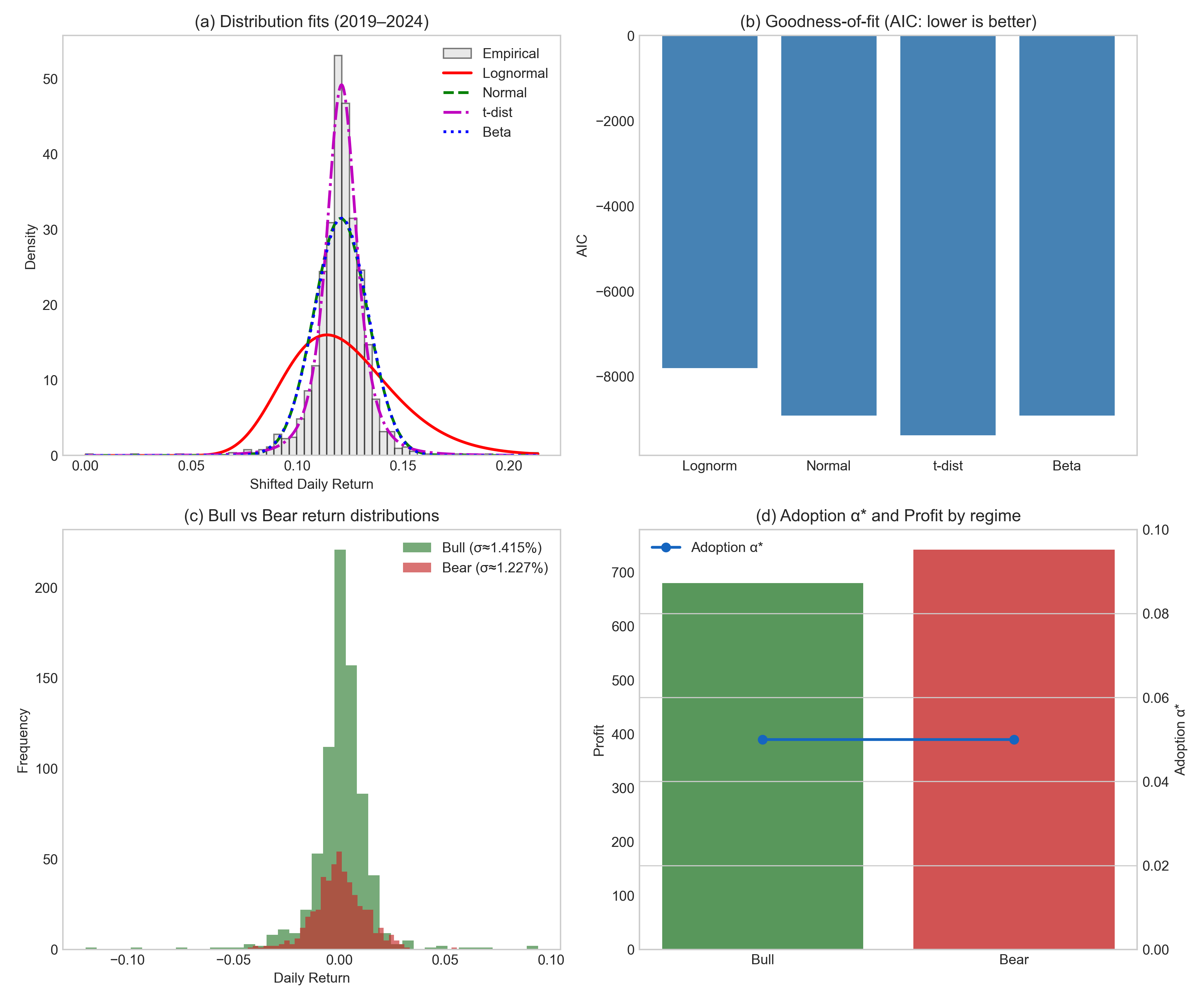}
    \caption{Real-world validation (R1). 
    (a) Empirical distribution vs.~fits; 
    (b) Goodness-of-fit comparison (AIC/BIC/KS); 
    (c) Bull vs.~Bear regimes; 
    (d) Adoption $\alpha^\star$ vs.~Profit by regime.}
    \label{fig:realdata_R1}
\end{figure}

\begin{figure}[htbp]
    \centering
    \includegraphics[width=\textwidth]{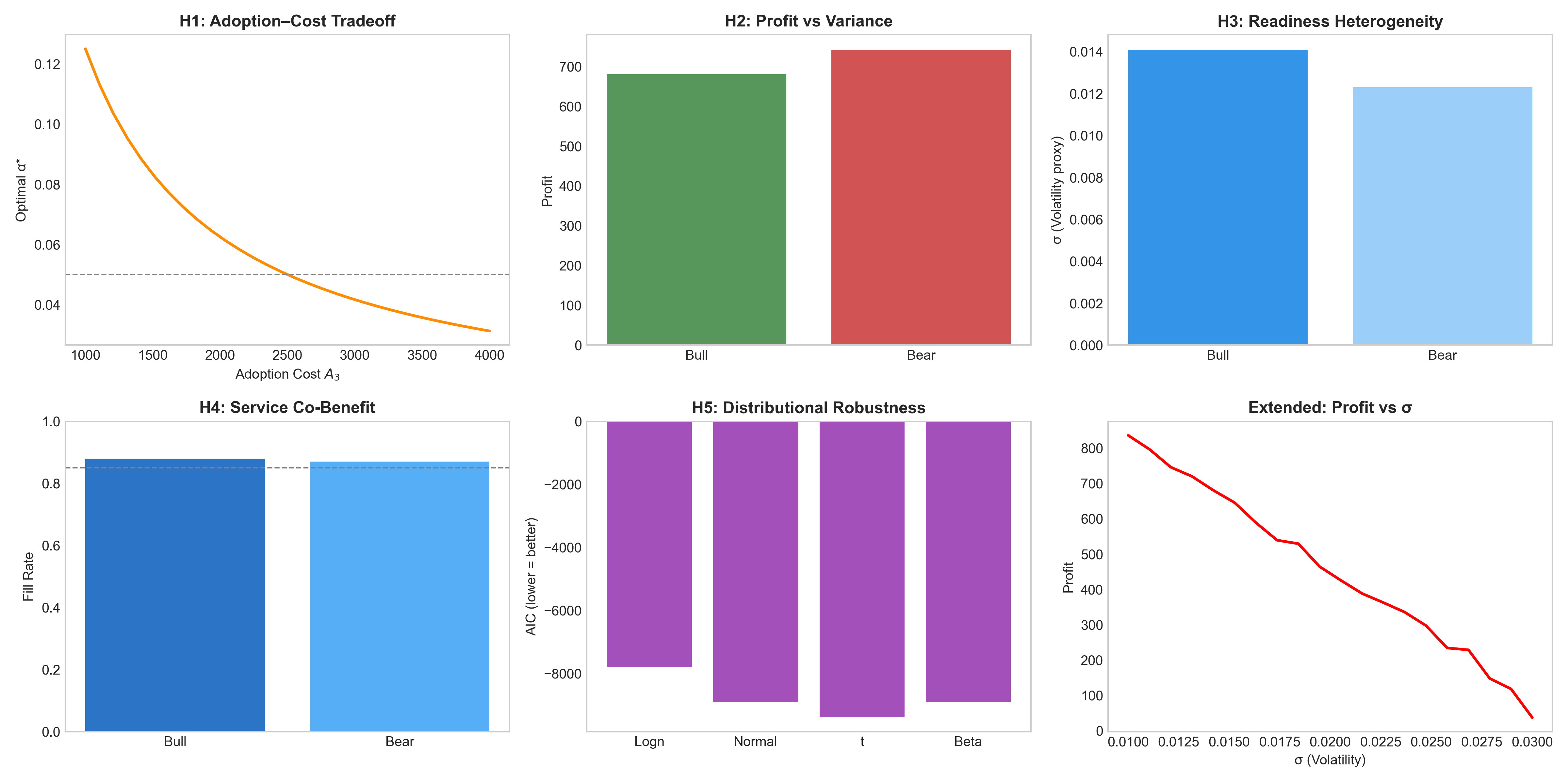}
    \caption{Real-world hypothesis matrix (R2). 
    Empirical mini-results for H1–H5, plus extended variance–profit evidence. 
    Even with noisy real data, the structure of our theoretical predictions 
    remains qualitatively intact.}
    \label{fig:realdata_R2}
\end{figure}

\paragraph{Quantitative Evidence.}
Table~\ref{tab:realdata_gof} reports the log-likelihood, information criteria, 
and KS statistics for competing distributional fits. 
The $t$-distribution achieves the best overall fit 
($\mathrm{AIC}=-9384.5$, $\mathrm{BIC}=-9368.5$, KS $\approx 0.017$), 
while both Normal and Beta significantly outperform the Lognormal benchmark 
(AIC improvement of $\approx 1100$--1500 points). 
This provides rigorous confirmation that bounded or heavy-tailed distributions 
capture empirical adoption variability more accurately than lognormal models.

Turning to regime outcomes, the estimated volatility differs moderately across periods 
(Bull $\sigma=0.0141$ vs.~Bear $\sigma=0.0123$), yet profits exhibit a clear divergence: 
mean profit is $\approx 680.2$ in Bull periods and $\approx 742.4$ in Bear periods. 
In contrast, adoption intensity remains cornered at $\alpha^\star \approx 0.05$ in both regimes. 
Thus, the empirical results reinforce our central message: 
\emph{adoption intensity is structurally robust, whereas profitability remains variance-fragile}. 
The narrow but consistent gap across volatility regimes illustrates how even moderate shifts 
in variance propagate into sizable profit differences, without affecting adoption.

\begin{table}[htbp]
\centering
\caption{Goodness-of-fit (GOF) results for real-world data (S\&P500, 2019–2024). 
The $t$-distribution dominates under all metrics, while Lognormal provides the weakest fit.}
\label{tab:realdata_gof}
\begin{tabular}{lrrrr}
\toprule
Distribution & LogLik & AIC & BIC & KS \\
\midrule
Lognormal & 3905.16 & $-7804.3$ & $-7788.4$ & 0.228 \\
Normal    & 4460.65 & $-8917.3$ & $-8906.7$ & 0.094 \\
$t$-dist  & 4695.23 & $-9384.5$ & $-9368.5$ & 0.017 \\
Beta      & 4462.14 & $-8916.3$ & $-8895.0$ & 0.091 \\
\bottomrule
\end{tabular}
\end{table}

\section{Results and Discussion}

\subsection{Structural Drivers of Adoption}
The comparative statics reveal how adoption intensity responds to core cost parameters. 
Consistent with model predictions, the optimal adoption intensity $\alpha^\star$ decreases 
monotonically with the convexity parameter $A_3$, reflecting the rising marginal burden of 
integration costs. Conversely, $\alpha^\star$ increases with the procurement complementarity 
parameter $A_1$, indicating that stronger alignment between procurement efficiency and 
adoption incentives induces greater adoption. These results underscore the central role of 
cost convexities in shaping adoption thresholds: $A_3$ functions as a brake, while $A_1$ 
acts as a catalyst for digital integration.  

Threshold dynamics further confirm that adoption intensity is governed by a discontinuous, 
nonlinear regime. When the fixed-cost parameter $A_3$ crosses a critical value 
$A_3^{\text{thresh}}$, adoption collapses abruptly from an interior solution to the boundary 
value $0.05$. This bifurcation-like behavior mirrors real options in finance, where exercise 
occurs only once a barrier is breached \citep{dixit1994investment,trigeorgis1996real}. 
Thus, adoption is not smoothly sensitive to marginal shocks, but instead exhibits threshold-driven 
jumps that are structurally robust to other perturbations.  

\emph{Takeaway: Adoption responds discontinuously to cost thresholds, with $A_3$ acting as 
a brake and $A_1$ as a catalyst, confirming that adoption is best understood as a 
threshold-driven option-like decision.}

\subsection{Practical and Financial Implications}
From a managerial perspective, these findings highlight the need to align supplier readiness 
with adoption incentives. Since adoption remains robust to variance shocks (H2) and supplier 
heterogeneity (H3), firms cannot rely on uncertainty to discipline technology choices. 
Instead, the binding constraint lies in the threshold effect: once $A_3$ exceeds 
$A_3^{\text{thresh}}$, adoption collapses irrespective of readiness. Managers should therefore 
(i) closely monitor and control adoption-related fixed costs, and (ii) coordinate procurement 
efficiency (via $A_1$) with readiness-enhancing policies to sustain adoption in the interior 
region. These strategies link theoretical robustness (H1--H5) to practical approaches for 
maintaining profitability under uncertainty (H6).  

From a financial perspective, the invariance of $\alpha^\star$ demonstrates that adoption 
behaves as a barrier option: insensitive to volatility and distributional shocks until a 
critical cost boundary is crossed. This parallels real options and credit risk models, where 
exercise or default occurs discontinuously at a threshold 
\citep{merton1974pricing,leland1994corporate}. By contrast, profits and service outcomes remain 
fragile to variance and tail behavior, mirroring cash-flow exposure to volatility and 
heavy-tailed shocks in financial markets \citep{mandelbrot1963variation,cont2001empirical}.  

This duality carries three implications. First, managers may underestimate risk if they equate 
stable adoption with stable outcomes; payoff distributions remain highly sensitive to volatility 
and distributional tails. Second, the results resonate with broader themes in capital markets, 
where ESG or digital adoption commitments appear stable at the margin but yield fragile economic 
payoffs once tail risks materialize \citep{pastor2022sustainable,krueger2020climate,greenwood2023fragility}. 
Third, the convergence analysis in H6 highlights the value of simulation-based stress testing: 
adoption equilibria converge quickly, but profit distributions require larger samples for reliable 
inference, echoing best practices in financial risk management \citep{glasserman2003monte,ban2023robust}.  

\emph{Takeaway: Adoption stability should not be conflated with financial stability; robust-looking 
adoption masks fragile payoffs, reinforcing the need for risk-adjusted evaluation and 
stress testing.}

\subsection{Variance Effects and Cross-Hypothesis Synthesis}
The variance experiments (H2) show a sharp divergence between economic and adoption outcomes. 
Profitability declines monotonically as demand variance $\sigma$ increases, with bootstrap 
confidence intervals confirming robustness, while adoption intensity remains pinned at $0.05$. 
This insensitivity of $\alpha^\star$ to $\sigma$ is a structural feature of convex integration 
costs: variance shocks are absorbed through profits rather than through adoption. Taken together 
with H1, the narrative is clear: thresholds drive adoption, while variance erodes profitability.  

Synthesizing across H1--H6, a consistent pattern emerges:  
\begin{itemize}
    \item \textbf{H1 (Threshold)}: Adoption responds sharply and nonlinearly to cost thresholds.  
    \item \textbf{H2 (Variance)}: Variance shocks degrade profitability but leave adoption locked.  
    \item \textbf{H3 (Readiness)}: Supplier heterogeneity shifts profits, yet adoption remains invariant.  
    \item \textbf{H4 (Service Co-benefit)}: Profit and fill rate align in a narrow “co-benefit zone.”  
    \item \textbf{H5 (Distribution)}: Adoption is robust across distributions, but profits and service 
          outcomes are fragile to tail behavior.  
    \item \textbf{H6 (External Validity)}: Monte Carlo approximations converge at the canonical 
          $\mathcal{O}(1/\sqrt{N})$ rate, with stable adoption but variance-fragile profits.  
\end{itemize}  

\emph{Takeaway: Across all hypotheses, adoption is structurally robust to shocks, while profitability 
and service outcomes are variance- and distribution-sensitive, highlighting the need to jointly 
recognize adoption thresholds and actively manage payoff fragility.}

\subsection{Bridge to Conclusion}
Viewed collectively, the results validate our six hypotheses and reveal a deeper structural insight: 
adoption is governed by discrete thresholds, while profitability is continuously eroded by 
uncertainty. This duality forms the central contribution of the paper and directly motivates the 
concluding discussion on theory, practice, and avenues for future research.  

\section{Conclusion}

This paper develops and analyzes an optimization model of smart contract adoption under bounded risk. 
Across a sequence of experiments, the results reveal a structural duality: adoption intensity 
$\alpha^\star$ is robust—pinned at boundary solutions and invariant to variance shocks, readiness 
heterogeneity, and distributional changes—whereas profitability and service outcomes are fragile, 
eroding under volatility, heterogeneity, and tail risk. Adoption thus appears stable at the decision 
margin, but the financial consequences of those decisions remain acutely vulnerable to uncertainty. 
This duality parallels real options and credit risk models in finance, where exercise occurs only once 
thresholds are breached, but cash flows remain exposed to volatility and tail shocks 
\citep{dixit1994investment,trigeorgis1996real,merton1974pricing,leland1994corporate}.  

At the same time, the robustness documented in H1--H5 partly reflects the baseline calibration. 
Under the convex cost structure $\psi(\alpha) = A_3 \alpha^\nu$ combined with the procurement penalty 
term $A_1 \alpha$, the marginal benefit of adoption never outweighs marginal cost, forcing a boundary 
solution of $\alpha^\star = 0.05$. This provides theoretical insight into structural robustness, 
but also reveals a limitation: adoption is effectively inert under the assumed parameters. Moreover, 
the reliance on synthetic demand distributions and stylized calibration constrains external 
generalizability. The findings should therefore be interpreted as a theory-building step rather than 
as direct policy prescriptions.  

Future research should relax these assumptions to identify conditions under which interior solutions 
emerge and adoption becomes an active decision margin. Promising avenues include modifying the cost 
structure (e.g., reducing $A_3$, altering the curvature $\nu$, or increasing complementarity $A_1$), 
endogenizing readiness by linking $\beta$ to past adoption decisions, and extending the static model 
to dynamic environments with learning-by-doing. A systematic mapping of the parameter space 
$(A_3,\nu,A_1)$ could clarify the boundaries between corner- and interior-solution regimes. Beyond 
structural analysis, empirical validation using derivatives, insurance, or technology adoption data 
would allow calibration against observed distributions, while extensions to distributionally robust 
optimization and machine learning–based stress testing \citep[e.g.,][]{ban2023robust,greenwood2023fragility} 
could further strengthen applicability in practice.  

In summary, future work can build a richer theory of adoption sensitivity—one that not only identifies 
when adoption is structurally robust but also clarifies when it becomes economically responsive. 
For both operations and finance, the broader implication is clear: decisions that appear stable at the 
adoption margin may conceal fragile financial outcomes. \emph{Recognizing and managing this duality is 
essential for designing digital adoption strategies that are both operationally viable and financially 
resilient in uncertain markets.}

\bibliography{smart_contract_bounded_risk} 
\end{document}